\documentclass{article}

\usepackage{fullpage}
\usepackage{amsmath}
\usepackage{amsfonts}
\usepackage{hyperref}
\usepackage{cleveref}
\usepackage[utf8]{inputenc}
\usepackage{tikz}
\usepackage{amsthm}
\usepackage{setspace}
\usepackage{frenchineq}
\usepackage[shortlabels]{enumitem}   
\newtheorem{theorem}{Theorem}
\newtheorem{assumption}{Asssumption}
\newtheorem{lemma}{Lemma}
\newtheorem{proposition}{Proposition}
\theoremstyle{remark}
\newtheorem{remark}{Remark}
\usetikzlibrary{calc}
\def\<#1,#2>{\langle #1,#2\rangle}
\newcommand{\control}{\mu}
\newcommand{\linear}{\mathcal{L}}
\newcommand{\dd}{\mathrm{d}}
\newcommand{\R}{\mathbb{R}}
\newcommand{\micu}{\mathrm{MICU}} %
\newcommand{\emt}{\mathrm{EMT}}  %
\newcommand{\adv}{\mathrm{adv}} %
\newcommand{\disp}{\mathrm{disp}} %
\newcommand{\phy}{\mathrm{GP}} %
\newcommand{\warning}{\mathrm{warn}} %
\newcommand{\alarm}{\mathrm{alarm}} %
\newcommand{\covid}{Covid}
\newcommand{\trop}{\mathrm{trop}}
\DeclareMathAlphabet{\mathbfcal}{OMS}{cmsy}{b}{n}

\DeclareMathAlphabet{\mathbbold}{U}{bbold}{m}{n}
 
\newcommand{\unit}{\mathbbold{1}} 

\usepackage[colorinlistoftodos,bordercolor=orange,backgroundcolor=orange!20,lineco
lor=orange,textsize=scriptsize]{todonotes}
\title{Understanding and monitoring the evolution of the Covid-19 epidemic from medical emergency calls: the example of the Paris area}
\author{Stéphane Gaubert$^{1,2}$, Marianne Akian$^{1,2}$, Xavier Allamigeon$^{1,2}$, Marin Boyet$^{1,2}$\\
  Baptiste Colin$^{1,2}$, Théotime Grohens$^{1,3}$,
  Laurent Massoulié$^{1,4,5}$,  David P. Parsons$^{1}$\\
  Frédéric Adnet$^{6,7}$,
  Érick Chanzy$^6$,
  Laurent Goix$^6$,
  Frédéric Lapostolle$^{6,7}$\\
  Éric Lecarpentier$^6$, Christophe Leroy$^6$,  Thomas Loeb$^6$, Jean-Sébastien Marx$^6$\\  
  Caroline Télion$^6$,   Laurent Tréluyer$^6$
  and Pierre Carli$^{6,8}$\\[1em]\small
   \\ \small
\\\small 
  $^1$ INRIA \\\small
  $^2$ CMAP, École polytechnique, IP Paris, CNRS\\\small
  $^3$ Université de Lyon, CNRS, INSA-Lyon, %
  LIRIS, UMR5205 \\\small
  $^4$ ENS, CNRS, PSL University\\\small
  $^5$ Microsoft Research-INRIA Joint Centre\\\small
  $^6$ AP-HP \\\small
  $^7$ Université Paris XIII, Bobigny\\\small
  $^8$ Université Paris-Descartes, Paris\\[1em]
  \small Emails: $^1$: {\small\texttt{Prenom.Nom@inria.fr}}\ \ \  $^6$: {\small\texttt{Prenom.Nom@aphp.fr}}
}
\date{July 20, 2020} %

\usepackage{siunitx} %

\usepackage{ragged2e}
\newcolumntype{L}[1]{>{\raggedright\let\newline\\\arraybackslash\hspace{0pt}}m{#1}}
\newcolumntype{C}[1]{>{\centering\let\newline\\\arraybackslash\hspace{0pt}}m{#1}}
\newcolumntype{R}[1]{>{\raggedleft\let\newline\\\arraybackslash\hspace{0pt}}m{#1}}

\usepackage{graphicx}

\newcommand{\cN}{\mathcal{N}}

\renewcommand{\th}{\textsuperscript{th}}
\newcommand{\st}{\textsuperscript{st}}
\newcommand{\nd}{\textsuperscript{nd}}
\newcommand{\rd}{\textsuperscript{rd}}
\begin{document}

\maketitle

{\bf Abstract.}
We portray the evolution of the Covid-19 epidemic during the crisis of
March-April 2020 in the Paris area, by analyzing the medical emergency
calls received by the EMS of the four central departments of this area
(Centre 15 of SAMU 75, 92, 93 and 94). Our study reveals strong
dissimilarities between these departments.  We show that the logarithm
of each epidemic observable can be approximated by a piecewise linear
function of time. This allows us to distinguish the different phases
of the epidemic, and to identify the delay between sanitary measures
and their influence on the load of EMS.  This also leads to an
algorithm, allowing one to detect epidemic resurgences. We rely on a
transport PDE epidemiological model, and we use methods from
Perron-Frobenius theory and tropical geometry.
  \\[1em]

  \begin{center}
    {\large  Comprendre et surveiller l'évolution de l'épidémie de Covid-19 à partir des appels au }\\[1mm] {\large numéro 15: l'exemple de l'agglomération parisienne}\\[1em]
    \end{center}

  {\bf Résumé.}
  Nous décrivons l'évolution de l'épidémie de Covid-19 dans
  l'agglomération parisienne, pendant la crise de Mars-Avril 2020, en
  analysant les appels d'urgence au numéro 15 traités par les SAMU des
  quatre départements centraux de l'agglomération (75, 92, 93 et 94).
  Notre étude révèle de fortes disparités entres ces
  départements. Nous montrons que le logarithme de
  toute observable épidémique peut être approché
  par une fonction du temps linéaire par morceaux.
  Cela nous permet d'identifier les différentes phases
  d'évolution de l'épidémie, et aussi d'évaluer
  le délai entre la prise de mesures sanitaires
  et leur effet sur la sollicitation de l'aide médicale urgente.
  Nous en déduisons un algorithme permettant de détecter
  une resurgence éventuelle de l'épidémie.
  Notre approche s'appuie sur un modèle d'EDP de transport
  de l'évolution épidémique, ainsi que sur des méthodes
  de théorie de Perron-Frobenius et de géométrie tropicale.

\section{Introduction}
The outbreak of \covid-19 in France has put the national Emergency Medical System (EMS), the \textit{SAMU}, in the front line.
In the \textit{Île-de-France} region, one most affected by the epidemic, the SAMU centers of Paris and its inner suburbs experienced a major increase in the number of calls received %
and of the number of ambulance dispatches for \covid-19 patients.

We show that indicators based on EMS calls and vehicle dispatches
allow to analyze the evolution of the epidemic.
In particular, we show that EMS calls
are early signals, allowing one to anticipate vehicle dispatch.
We provide a method of short term prediction of the evolution
of the epidemic, based on mathematical modeling. 
This leads to
{\em early detection and early alarm mechanisms}
allowing one either to confirm that certain sanitary
measures are strong enough to contain the epidemic,
or to detect its resurgence.
These mechanisms
rely on simple data generally available in EMS:  numbers
of patient records tagged as \covid-19, and among these,
numbers of records resulting in medical advice,
ambulance dispatch, or Mobile Intensive Care Unit
dispatch.
We also provide a comparative description of the evolution
of the epidemic in the four central departments of the
Paris area, showing spatial dissimilarities,
including a strong variation of the doubling time,
depending on the department. 

Our approach relies on several mathematical tools
in an essential way. Indeed, the Covid-19 epidemic has
unprecedented characteristics, and, given the lack of experience
of similar epidemics, one needs to rely on mathematical models.
We use transport PDE to represent the dynamics
of Covid-19 epidemic. Transport PDE
capture epidemics with a significant time interval between
contamination and the start of the infectious phase
(in contrast, ODE models without time delays allow
instantaneous transitions from contamination to the infectious phase).
In the early stage of the epidemic, in which the majority of the
population is susceptible, this dynamics
becomes approximately linear and order preserving.
Then, it can be analyzed by methods of Perron--Frobenius
theory. Our main theoretical result shows
that the logarithm of epidemic observables
can be approximated by a piecewise linear map, with
as many pieces as there are phases of the epidemic
(i.e., periods with different contamination conditions),
see~\Cref{th-1}. This methods allows us to identify,
the phases of the epidemic evolution, and also to evaluate the
time interval between sanitary measures
and their impact on epidemic observables,
like vehicle dispatch.
The idea of piecewise linear approximation and of ``log glasses'',
a key ingredient of the present approach, arises
from tropical geometry.

The present work started on March 13\th, and
led to the algorithm presented here.
A preliminary version of this algorithm was used,
on March 20\th, to forecast the epidemic
wave, anticipating that the peak load of SAMU (which occurred around March 27\th) would be different depending on the department of the Paris area.
We subsequently applied our method to provide Assistance Publique -- Hôpitaux de Paris (AP-HP), on April 5\textsuperscript{th}, with an early report,
quantifying the efficiency of the lockdown measures
from the estimation of the contraction rate of the epidemic
in the different departments.
This algorithm is now deployed operationally
in the four SAMU of AP-HP. This work may be quickly reproduced
in any EMS.

Although it was developed for Covid-19 and for EMS calls, the present monitoring method is generic. It may also apply to other medical indicators, see \Cref{ssec-alarm}, and to other epidemics, for instance, influenza.

This paper is a crisis report, giving a unified picture of a
work done jointly by a team of physicians
of the SAMU of AP-HP and applied mathematicians from INRIA
and École polytechnique. Medical, epidemiological, and mathematical aspects
are intricated in this work. We received help
from several physicians, researchers and engineers,
not listed as authors, and also help from several
organizations. They are thanked in the acknowledgments section.

This paper should be understood as an announce. The results will be
subsequently developed in several papers, with different subsets of coauthors.
It is intended to be read both by a medical and a mathematical
audience. The first part of the paper, up to \Cref{sec-discuss}
included, and the conclusion, are intended to a broad audience.
Mathematical tools are presented in~\Cref{sec-pde}, \Cref{sec-tropical},
\Cref{sec-proba} and in the appendix.

The present work shows the epidemiological significance of the calls received by the EMS, it focuses on the mathematical modeling aspects, on the description
of the evolution of the epidemy in the Paris area, and on prediction algorithms. The current work\footnote{COVID19 APHP-Universities-INRIA-INSERM, 
Emergency calls are early indicators of ICU bed requirement during the COVID-19 epidemic,  medRxiv:2020.06.02.20117499, June 2020.}
with an intersecting set of authors, is coordinated with the present one. It focuses on medical aspects. It makes a case study of  the Covid-19 crisis of March-April 2020, in Paris, considering the EMS and the hospital services in a unified perspective.  It shows that the calls received by SAMU are early predictors of the future load on ICU. 

\section{Context}

The mission of the SAMU centers is to provide an appropriate
response to calls to the number 15,
the French toll-free phone number dedicated to medical emergencies. This service is based on the medical regulation of emergency calls, in the sense that for each patient, a physician decides which response is most appropriate. Thus, depending on the evaluation
over the phone of the severity of the case and the circumstances, the response may be a medical advice, a
home visit by a general practitioner,
the dispatch of a team of EMTs (Emergency Medical Technicians) of either a first aid association or the Fire brigade, or an ambulance of a private company. A Mobile Intensive Care Unit (MICU), staffed by a physician, a nurse and an EMT, is sent to the scene as a second or a first tier, when a life threatening problem is suspected. The role of the SAMU in the management of %
disasters or mass casualties has been described elsewhere~\cite{hirsch2015medical,baker2007multiple}. 
The city of Paris and its inner suburbs are covered by 4 departmental SAMU Center-15 : Paris (75), Hauts-de-Seine (92), Seine Saint-Denis (93), and Val-de-Marne (94), see the map on \Cref{fig-all}. They serve a population of 6.77 million inhabitants. These four Center-15 are part of the public hospital administration, AP-HP (Assistance Publique -- Hôpitaux de Paris). They operate identically and use the same computerized call management system. 
Since the outbreak of the \covid-19 epidemic, the French government
instructed the public that anyone with signs of respiratory infection or fever should not go directly to the hospital emergency room to limit overcrowding, but should call number 15 for orientation.
To comply with the recommendations of the health care authorities,
the four Center-15 applied the same procedures: after medical call regulation, only patients with signs of severity or significant risk factors were transported by EMTs and ambulances to hospitals, either to Emergency Room (ER) or newly created \covid-19 Units. The cases presenting a life-threatening emergency, mostly respiratory distress, were managed by a MICU team and then admitted directly in Intensive Care Unit (ICU).
All other cases were advised to stay at home and isolate themselves. When necessary, these patients were also eligible for a home visit by a general practitioner or a consultation appointment the following days.

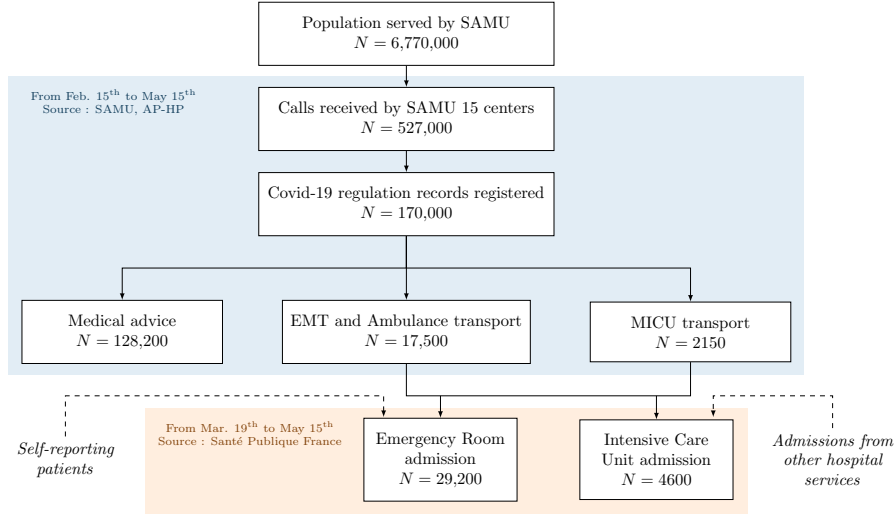
\begin{figure}[htbp]
  \begin{center}
\resizebox{.75\textwidth}{!}{
\begin{tikzpicture}
\tikzset{place/.style={draw,circle,inner sep=2.5pt,semithick}}  
\definecolor{bluePY}{HTML}{1F77B4} %
\definecolor{orangePY}{HTML}{FF7F0E} %

  \tikzset{rectangle-box/.style={rectangle, draw, fill=white, inner ysep=.3cm, minimum width=3cm}}

  \def\px{6}
  \def\py{1.8}

  \fill[bluePY,   opacity=.13] (-1.4*\px, -3*\py) -- ++(2.8*\px, 0) -- ++(0,3.5*\py) -- ++(-2.8*\px,0) -- cycle;
  \fill[orangePY, opacity=.13] (-.92*\px, -4.65*\py) -- ++(2.12*\px, 0) -- ++(0,1.25*\py) -- ++(-2.12*\px,0) -- cycle;

  \node[color=bluePY!50!black] (blue-legend) at (-1.03*\px, .2*\py) {\begin{minipage}{5cm}\centering\scriptsize From Feb. 15\textsuperscript{th} to May 15\textsuperscript{th} \\ Source : SAMU, AP-HP\end{minipage}};
  \node[color=orangePY!50!black] (or-legend) at (-.55*\px, -3.7*\py) {\begin{minipage}{5cm}\centering\scriptsize From Mar. 19\textsuperscript{th} to May 15\textsuperscript{th} \\ Source : Santé Publique France\end{minipage}};

  \node[rectangle-box] (population) at (0,\py)         {\begin{minipage}{6cm}\centering Population served by SAMU            \\ $N=\num[group-separator={,}]{6 770 000}$ \end{minipage}};
  \node[rectangle-box] (calls)      at (0,0)           {\begin{minipage}{6cm}\centering Calls received by SAMU 15 centers    \\ $N=\num[group-separator={,}]{527 000}$   \end{minipage}};
  \node[rectangle-box] (drm-cov)    at (0,-\py)        {\begin{minipage}{6cm}\centering \covid-19 regulation records registered \\ $N=\num[group-separator={,}]{170 000}$   \end{minipage}};
  \node[rectangle-box] (cov-micu)   at (\px,-2.5*\py)  {\begin{minipage}{4cm}\centering MICU transport                       \\ $N=\num[group-separator={,}]{2 150}$     \end{minipage}};
  \node[rectangle-box] (cov-emt)    at (0,-2.5*\py)    {\begin{minipage}{5cm}\centering EMT and Ambulance transport          \\ $N=\num[group-separator={,}]{17 500}$    \end{minipage}};
  \node[rectangle-box] (cov-advice) at (-\px,-2.5*\py) {\begin{minipage}{4cm}\centering Medical advice                       \\ $N=\num[group-separator={,}]{128 200}$   \end{minipage}};
  \node[rectangle-box] (er-covid)   at (.12*\px,-4*\py)      {\begin{minipage}{3cm}\centering Emergency Room admission    \\ $N=\num[group-separator={,}]{29 200}$    \end{minipage}};
  \node[rectangle-box] (icu)        at (.88*\px,-4*\py)    {\begin{minipage}{3cm}\centering Intensive Care Unit admission    \\ $N=\num[group-separator={,}]{4 600}$     \end{minipage}};

  \node (other_arr) at (1.5*\px,-4*\py)    {\begin{minipage}{3cm}\centering \textit{Admissions from other hospital services}  \end{minipage}};
  \node (self)      at (-1.2*\px,-4*\py)    {\begin{minipage}{3cm}\centering \textit{Self-reporting patients}  \end{minipage}};

  \draw[->,>=latex] (population) -- (calls);
  \draw[->,>=latex] (calls) -- (drm-cov);
  \draw[->,>=latex] (drm-cov) |- ($(drm-cov)-(0,.75*\py)$) -| (cov-micu);
  \draw[->,>=latex] (drm-cov) |- ($(drm-cov)-(0,.75*\py)$) -| (cov-emt);
  \draw[->,>=latex] (drm-cov) |- ($(drm-cov)-(0,.75*\py)$) -| (cov-advice);
  \draw[>=latex] (cov-emt) -- ($(cov-emt)-(0,.75*\py)$) -| (cov-micu);
  \draw[->,>=latex] ($(er-covid) + (0,.75*\py)$) -- (er-covid);
  \draw[->,>=latex] ($(icu) + (0,.75*\py)$) -- (icu);

  \draw[->,>=latex, dashed] (self) |- ($(self)+(.5*\px,.75*\py)$) -| ($(er-covid.center) + (-.2*\px,.5*\py)$);
  \draw[->,>=latex, dashed] (other_arr) |- ($(other_arr)+(-.1*\px,.75*\py)$) -| ($(icu.center) + (.2*\px,.5*\py)$);

\end{tikzpicture}
}
\caption{Flowchart: from calls to Center 15 to admission in hospital units. The numbers are summed over the departments 75, 92, 93 and 94 of the Paris area.}
\label{fig-flow}
\end{center}
\end{figure}

In order to maintain a rapid response when a major increase in the number of calls was observed, 
the four Center-15 implemented specific procedures.
Switchboard operators and medical staff was reinforced, and for calls related to Covid-19 an interactive voice server ---triaging the calls to dedicated computer stations--- was developed.
Patient evaluation and management were improved by introducing video consultation, sending of instruction using SMS, giving the patient the option to be called back. Prehospital EMT teams were also significantly reinforced by first aid volunteers,
and additional MICU were created.
Since January 20\textsuperscript{th} 2020 all calls and patient records related to \covid-19 were flagged in the information system of Center-15 and a daily automated activity report was produced.

\section{Methods}\label{sec-methode}
In this section, we describe the methods used in this work, in a way
adapted to a general audience. Mathematical developments appear
in~\Cref{sec-pde,sec-tropical,sec-proba} and in the appendix. 

\subsection{Classification of calls}\label{sec-classify}
In order to develop a mathematical analysis of the evolution of the epidemic,
we classified the calls tagged as \covid-19
in three categories, according to the decision taken:
\medskip
\begin{itemize}[align=left]
\item[Class 1:]  calls resulting in the dispatch of a Mobile Intensive Care Unit;%
\item[Class 2:]  calls resulting in the dispatch of an ambulance staffed with EMT;%
\item[Class 3:]  calls resulting in no dispatch decision. Such calls
  correspond to different forms of medical advice (recommendation to consult a GP, specific instructions to the patient, etc.).
\end{itemize}
\medskip

We shall denote by $Y_{\micu}(t)$
(resp.\ $Y_{\emt}(t)$ and $Y_{\adv}(t)$)
the number of MICU transports
(resp.\ the number of ambulances transport and
the number of medical advices)
on day $t$, for patients tagged with suspicion of \covid-19. We shall
call these functions of time the {\em observables}, in contrast
with $C(t)$, the actual number of new contaminations on day $t$, which cannot be
measured.
We developed a piece of software that computes these observables
by analyzing the medical decisions associated with the patient records, made accessible daily by AP-HP.

\subsection{Mathematical properties of the observables}\label{sec-model}
To analyze these observables, we shall rely on a mathematical
model.

A standard approach represents the evolution
of an epidemic by an ordinary differential equation
(SEIR ODE), representing the evolution of the population
in four compartments: ``susceptible'' (S),
``exposed'' but not yet infectious (E), ``infectious'' (I),
and finally, ``removed'' from the contamination
chain (R), either by recovery or death.  A
refinement of the SEIR model splits the S and E compartments
in sub-compartments corresponding to different age classes.
It includes a contact matrix, providing differentiated
age-dependent infectiosity rates~\cite{crepey}.
Another refinement includes additional
compartments, representing, for instance, patients
at hospital~\cite{colizza}, or 
individuals with mild symptoms~\cite{magalbiology}.

In contrast with such ODE models,
we use a partial differential equation (PDE), i.e., an infinite
dimensional dynamical system, described in~\Cref{sec-pde}.
Our approach is inspired by the
PDE model of Kermack and McKendrick~\cite{kermack-mckendrick}.
We use PDE, rather than ODE, to take into account
the presence
of {\em delays} in the contamination process:
the median incubation period of
Covid-19 is estimated of 5.1 days, with a
95\% confidence interval of 4.5-5.8,
and 95\% of patients in the range [2.2,11.5], see~\cite{laueretal},
in line with
other human coronaviruses having also long
incubation times, like SARS~\cite{varia} and MERS~\cite{virlojeux}.
(This may be compared with a median incubation time 
of 1.4 days [95\% CI, 1.3--1.6]
for the toxigenic Cholera~\cite{cholera},
or with an interval of 36 hours
between infection by pneumonic plague and first symptoms 
in Brown Norway rats, with rapid letality, 2-4 days after infection~\cite{plaguerats}.)
ODE models assume exponentially distributed transitions times
from one compartment to another. This entails
that the interval elapsed between contamination and the
time an individual becomes infectious can be arbitrarily
small, so ODE models are more adapted to epidemics
with a short incubation time. 
Using transport PDE, as is done in~\Cref{sec-pde}, takes
delays into account, allowing also one to recover
ODE models as special cases.

We shall limit here our analysis to the early stage
of the epidemic, assuming that
the population that has been infected is much
smaller than the susceptible population.
This approximation is reasonable at least in the initial
part of the epidemic,
according to the study~\cite{salje} which gives
an estimate of 5.7\% for the proportion
of the population in France that has been infected
prior to May 11\th, 2020. Then, the dynamics
becomes linear and order-preserving. The latter
property entails that the observables are an increasing function
of the size of the initial population that is either exposed or infected.

Results of Perron--Frobenius and of Krein-Rutman theory,
which we recall in \Cref{sec-epidemio}, entail
that, {\em if the sanitary measures stay unchanged},
there is a rate $\lambda$,
such that
the number of newly contaminated individuals
at day $t$ grows as $C(t)\simeq K_C \exp(\lambda t)$,
as $t\to\infty$, where $K_C$ is a positive constant.

The number $\delta:= (\log 2)/\lambda$, when it is
positive, represents the {\em doubling time}:
every $\delta$ days, the number of new 
contaminations per day doubles.
When $\delta$ is negative, the epidemic is in a phase
of exponential decay. Then, the opposite of $\delta$
yields the time after which the number of new
contaminations per day is cut by half.
For the analysis which follows, it is essential
to consider, instead of $C(t)$, its logarithm,
$\log C(t)\simeq \log K_C + \lambda t$.
The exponential growth or decay of $C(t)$
corresponds to a linear growth or decay
of the logarithm.

We shall also see in~\Cref{sec-pde} that all the epidemiological observables
evolve with the same
rate. E.g., assuming that all the patients transported
by MICU were contaminated $\tau_{\micu}$
days before the transport,
and that a proportion $\pi_{\micu}$ of the contaminated individuals
will require MICU transport, we arrive at $Y_{\micu}(t) = \pi_{\micu}C(t-\tau_{\micu})$,
and so $\log Y_{\micu}(t) \simeq \log \pi_{\micu}+ \log K_C + \lambda(t-\tau_{\micu})$.
Similar formul\ae\ apply to $Y_{\emt}$ and $Y_{\adv}$,
and to other observables based for instance on
ICU admissions or deceases. A finer
model of observables, taking into account a distribution of times $\tau_{\micu}$,
  instead of a single value, is presented in~\Cref{subsec-observables}.

For the analysis which follows, we shall keep in mind that
{\em the logarithm of all the observables is asymptotically
linear as $t\to \infty$}, and that {\em the rate, $\lambda$,
is independent of the observable}. 

\subsection{Piecewise linear approximation of the logarithm of the observables}
\label{sec-piecewise}
When the sanitary measures change, for instance,
when lockdown is established, the rate $\lambda$
changes. So, the logarithm
of the observables cannot be approximated
any more by a linear function.
However, a general result, stated as \Cref{th-1} below,
shows that this logarithm can be approximated
by a {\em piecewise linear function} with
as many linear pieces as there are
phases of sanitary policy.
This result stems from the order preserving
and linear character of the epidemiological
dynamics, and so, it holds for a
broad class of epidemiological models; several
examples of such models are discussed in~\Cref{sec-pde}.

In the Paris area, there are three relevant sanitary
phases to consider from February to May, 2020: 
initial growth (no restrictions); ``stade 2'' (stage 2)
starting on Feb.\ 29\textsuperscript{th} (prevention measures), and then lockdown from March 17\textsuperscript{th}
to May 11\textsuperscript{th}. Sanitary phases are further described
in~\Cref{sec-delay}.

Since the number $\nu$ of sanitary phases is known
(here $\nu=3$), we can infer the different values of $\lambda$ attached
to each of these phases, 
by computing the best piecewise linear approximation, $\linear(t)$
with at most $\nu$ pieces of the logarithm of an observable $Y(t)$.
To compute a robust approximation, we minimize
the $\ell_1$ norm, $\sum_{t} |\linear(t)- \log Y(t)|$,
where the sum is taken over the days $t$
in which the data are available.
Finding the best approximation $\linear$ is
a difficult optimization problem, for
the objective function is both non-smooth and non-convex.
Methods to solve this problem are
discussed in~\Cref{appendix-2}.
\subsection{Epidemic alarms based on doubling times}
\label{ssec-alarm}
To construct epidemic alarms, we shall compute
a linear fit,  $\linear(t)=\alpha+\beta t$,
to the variables $\log Y(t)$, where $Y$ is an epidemic observable.
The principle is to trigger an alarm
when the doubling time becomes positive, or equivalently, when the slope $\beta$
becomes positive.

Assuming that values of $Y(t)$ are known over a temporal window,
there are simple ready-to-use methods for computing estimates $\hat{\beta}$ for the slope $\beta$. 
We can also determine the probability $p^+$
that the slope is positive.
These methods are detailed in~\Cref{sec-appendixC}.
On their basis, we propose the following {\bf alarm raising mechanism,
  allowing one to deploy a gradual response}.

This mechanism relies on the two following observables,
$Y_{\adv}$, the number of calls resulting in medical advice,
and $Y_{\disp}:=Y_{\emt}+Y_{\micu}$, the number of dispatched vehicles.
The consolidation of the observables $Y_{\emt}$ and $Y_{\micu}$
is justified, because the two time
series both correspond to the stage of aggravation, albeit with different
degrees, and so they evolve more or less at the same time.

First define a temporal window of days $t$ over which the linear fit
$\linear_{\adv}(t)=\alpha_{\adv}+\beta_{\adv} t$ to $\log Y_{\adv}(t)$ is made.
By default we consider the last ten days prior to the current day.
Similarly,  we compute a linear fit
$\linear_{\disp}(t)=\alpha_{\disp}+\beta_{\disp} t$ to $\log Y_{\disp}(t)$
over the same time window.

Our algorithm will generate both a {\em warning} and
{\em alarms}. A warning is a mere incentive to be careful.
An unjustified warning is bothersome but generally harmless,
so we accept a high probability of false positive
for warnings. An alarm may imply some actions,
so we wish to avoid false alarms.
For this reason, we shall consider two different probability thresholds,
$\vartheta_\alarm$ and $\vartheta_\warning$, say
$\vartheta_\alarm = 75\%$ and $\vartheta_\warning=25\%$.
With this setting, we will be warned as soon as the probability
of the undesirable event is $\geq 25\%$, and we will be alarmed
when the same probability becomes $\geq 75\%$.
Of course, these thresholds can be changed, depending
on the risk level deemed to be acceptable.
We shall denote by $p^+_\adv$ the probability
that the slope $\beta_{\adv}$ is positive,
and by $p^+_{\disp}$ the probability that $\beta_{\disp}$
is positive.
These probabilities are evaluated on the basis of statistical assumptions detailed in \Cref{sec-appendixC}.

\begin{enumerate}
\item A {\bf warning} is provided when $p^+_\adv\geq \vartheta_\warning$,
  meaning that the probability that the slope
  $\beta_\adv$   of the curve of the logarithm of the
  {\em calls for medical
advice} over the corresponding time window  
  be positive is at least $\vartheta_\warning$.
  This should be interpreted as a mere warning
  of epidemic risk: choosing $\vartheta_\warning$ as above,
  the odds are at least 25\% that the epidemic is growing.
\item  This warning is subsequently transformed
  into an {\bf alarm} when $p^+_{\adv}\geq \vartheta_\alarm$.
  Choosing $\vartheta_\alarm$ as above, the odds
  that the epidemic is growing are now
at least $75\%$.
\item Such an alarm is then subsequently transformed into a
  {\bf confirmed alarm}
  if we still have $p^+_{\adv}\geq \vartheta_\alarm$, and if, in addition,
  $p^+_{\disp}\geq \vartheta_\alarm$, meaning that the probability
  that the slope of the logarithm of the curve of ambulances and MICU dispatches be positive is now above $\vartheta_\alarm$.
  Again, this estimate is defined in terms of a time window over which $\beta_{\disp}$ is estimated.
  We use the same default values of ten days and $\vartheta_\alarm$ as above.
\end{enumerate}
As shown in~\Cref{sec-data}, the indicators based on vehicle
dispatch are by far less noisy than the indicators based
on calls for medical advices, but their evolution
is delayed. This is the rationale
for using medical advice for an early warning and early alarm,
and then vehicle dispatch for confirmation.

Instead of considering the probability $p^+$,
we could consider the upper and lower bounds of a confidence
interval $[\beta^-_\epsilon,\beta^+_\epsilon]$
for the estimated slope $\beta$,  with a probability threshold $\epsilon$.
Then we may, trigger a warning when $\beta^+_\epsilon\geq 0$, and an alarm
when $\beta^-_\epsilon\geq 0$. This leads to an essentially equivalent
mechanism. We prefer the algorithm above as it allows to
interpret the thresholds in terms of false positives and false negatives.

Given the severity of the risk implied by Covid-19,
it may be desirable to complete the previous alarm, based
only on tail probabilities of the slope, by a different type of alarm,
based on a threshold of doubling time, $D$.
The alarm will be triggered if the odds
that the doubling time be positive and smaller than $D$
are at least one half. An indicative value of $D$
might be 14 days: a doubling of the number of
arrivals of Covid-19 patients in hospital services 
every 14 days may be quite challenging, justifying
an alarm, and the slope corresponding to this
doubling time seems significant enough to avoid
false alarms. Again, the value of $D$
can be changed arbitrarily depending
on the acceptable level of risk.
Moreover, this other type of alarm
can still be implemented in two stages: early
alarm, with the medical advice signal,
and then confirmed alarm, with the vehicle dispatch signal.

In addition, \Cref{sec-appendixC} provides more sophisticated ready-to-use methods for obtaining sharper confidence intervals or probabilities
for the slope $\beta$,
resulting in more precise alarm mechanisms,  when different time series
are available. We require, however,
that these series correspond to events occurring approximately
at the same stage in the pathology unfolding.
Here, we used the trivial aggregator,
$Y_{\disp}=Y_{\emt}+Y_{\micu}$. There is an optimal way to mix different
series to minimize the variance of the composite estimator, explained
in~\Cref{sec-appendixC}.

This methodology is generic. It could thus also apply to obtain a sharper confidence interval for the early indicator by combining its estimate $\hat{\beta}_{\adv}$ with that of other time series associated with signals that correspond to the same stage in pathology unfolding. Specifically, the count $Y_{\phy}(t)$ of patients consulting general practitioners  for recently developed \covid-19 symptoms, if available, provides such a signal. A linear fit to $\log Y_{\phy}(t)$ would then yield an estimate $\hat{\beta}_{\phy}$ which can be combined with $\hat{\beta}_{\adv}$ to refine the corresponding confidence interval. In this way, we can
mix several early but noisy indicators to get an
early but less noisy consolidated indicator.

\section{Results -- data analysis}\label{sec-data}
\subsection{Key figures and graphs}
From February 15\th~to May 15\th, we counted 
a total of \num[group-separator={,}]{170166} patient files tagged with a suspicion of \covid-19,
distributed as follows in the different departments:
\num[group-separator={,}]{53646} in Dep.~75; \num[group-separator={,}]{36721} in Dep.~92;
\num[group-separator={,}]{49703} in Dep.~93; and \num[group-separator={,}]{30096} in Dep.~94.

The flow of calls to the SAMU of the Paris area, and its impact on ER and ICU,
is shown on Figure~\ref{fig-flow}. The data concerning the ER and the ICU
are taken from the governmental website SPF (Santé Publique France)~\cite{SPF},
it is available only from March 19\th.

On Figure~\ref{fig-differenttypes}, we represent, in logarithmic ordinates,
the numbers of events of different types, summed over the four departments
of the Paris area (75, 92, 93 and 94):
(i) the number of patients calling the SAMU (including 
patients not calling for \covid-19 suspicion);
(ii) the number of calls tagged as \covid-19 not resulting in a vehicle
dispatch (i.e., as discussed in \S\ref{sec-classify}, all kinds of medical advices);
(iii) the number of calls tagged as \covid-19
resulting in an ambulance or MICU dispatch,

We obtained the data (i) by analyzing the phone operator
log files. Since a patient may call the Center 15 several times,
we eliminated multiple calls to count unique patients.
To compute data (ii) and (iii), we developed
a software to analyze the ``medical decision'' field
of the regulation records.

Using logarithmic ordinates is essential on~\Cref{fig-differenttypes},
as it allows to visualize on the same graph signals
of different orders of magnitude (e.g, there is a ratio of 20 between
the peak number of patients calling and the peak number
of vehicles dispatched).

\begin{figure}[htbp]
  \begin{center}
    \includegraphics[scale=0.42,trim={1cm, 1cm, 1cm, 1cm}, clip]{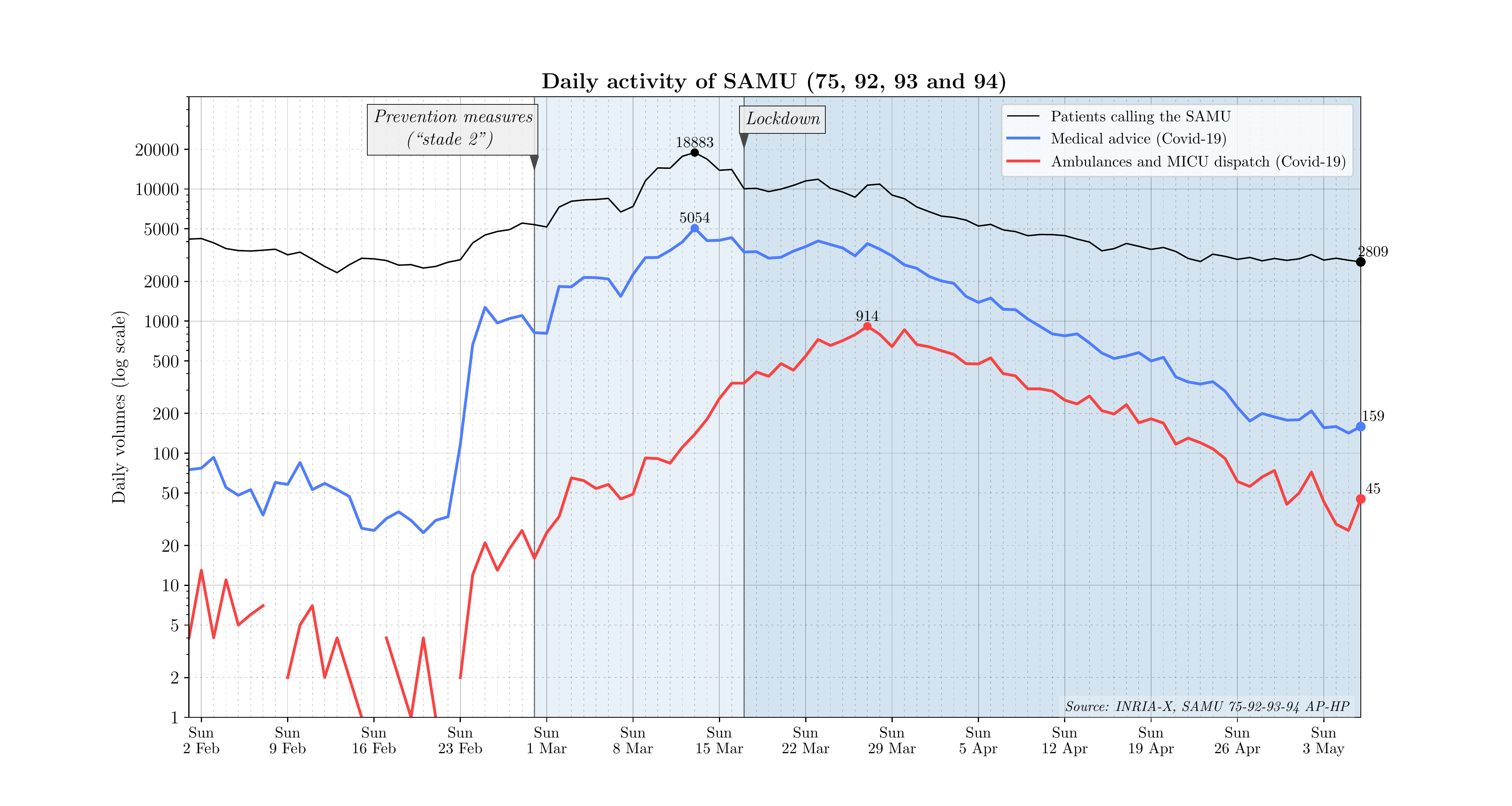}
  \end{center}
  \caption{Number of patients calling Center-15, of MICU and ambulances dispatch for \covid-19 suspicion
   in the Paris area (departments 75, 92, 93 and 94)}
    \label{fig-differenttypes}
\end{figure}

The evolution of the number of vehicles dispatched (MICU and ambulances)
is shown on Figure~\ref{fig-all}, for each department of the Paris area
(still with logarithmic ordinates).

\begin{figure}[htbp]
  \begin{center}
    \begin{tikzpicture}
      \node[anchor=south west,inner sep=0] at (0,0) {
        \includegraphics[scale=0.57,trim={2cm, 1cm, 0cm, .3cm}, clip]{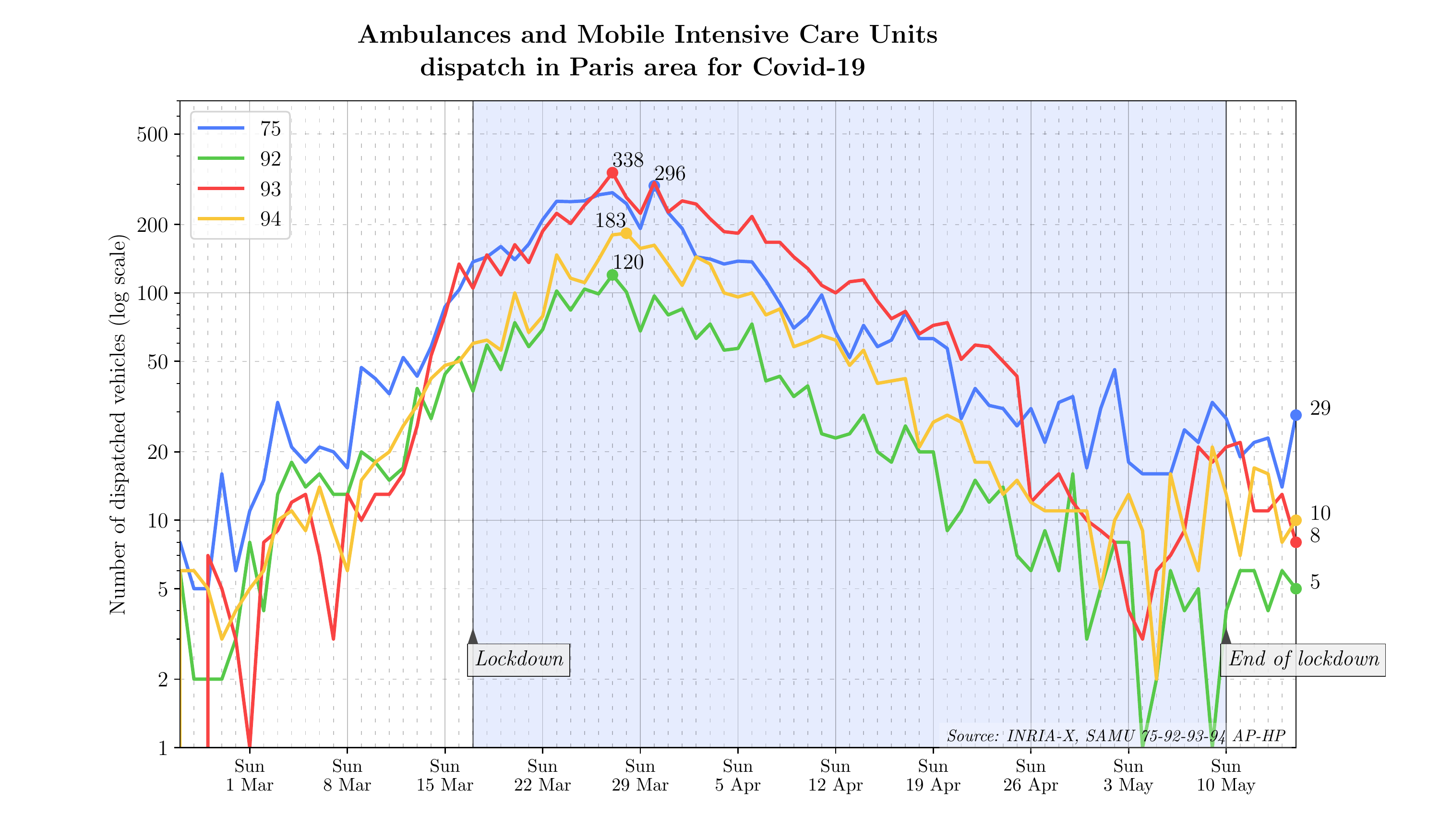}
    };
    \node[anchor=south west,inner sep=0] at (10.3,6.1) {
      \includegraphics[scale=0.4,trim={0cm, 0cm, 0cm, 0cm}, clip]{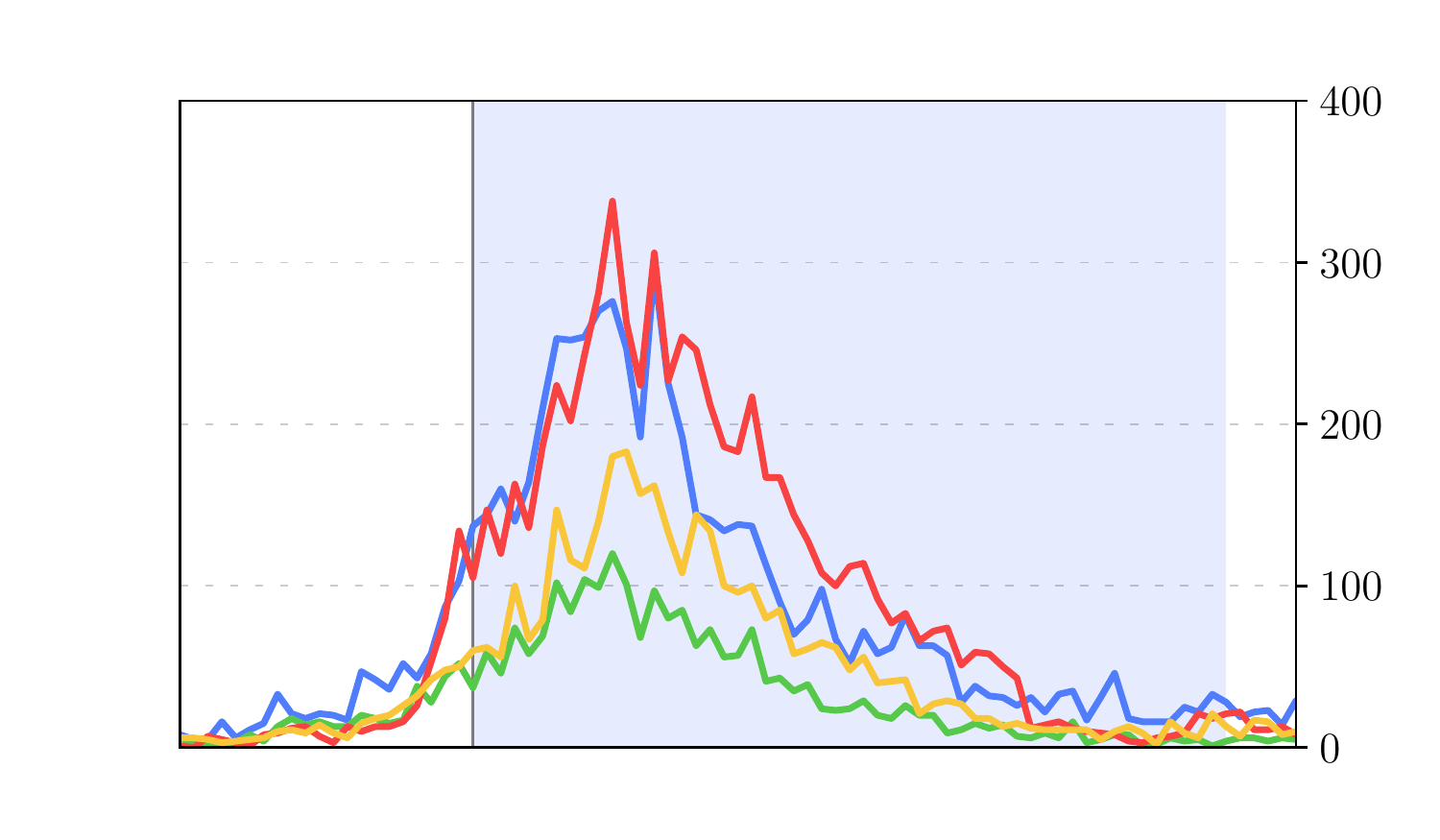}
  };

    \end{tikzpicture}
  \end{center}
  \vspace{-.4cm}
  \caption[Departmental comparison of vehicles dispatch]{ 
\begin{tabular}{L{12cm}C{1.8cm}} Comparison of the evolution of the epidemic in the different departments of the Paris area: numbers of vehicles dispatch by department. The figure inset displays the same curves in usual linear ordinates to keep in mind the different magnitudes at stake. A map of the Paris area, showing the departments 75, 92, 93, 94, is at the bottom right of the figure. & 
    \vspace{-.3cm}
    \begin{tikzpicture}
      \node[opacity = .6] at (0,0) {\includegraphics[scale = .33]{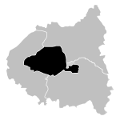}};
      \node[opacity = .6] at (0,0) {\includegraphics[scale = .33]{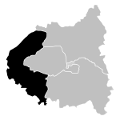}};
      \node[opacity = .6] at (0,0) {\includegraphics[scale = .33]{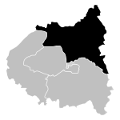}};
      \node[opacity = .6] at (0,0) {\includegraphics[scale = .33]{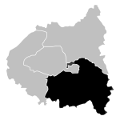}};
      \node[white] at (-0.12,-0.02) {\tiny 75};
      \node[black] at (-0.57,-0.45) {\tiny 92};
      \node[black] at ( 0.1,  0.6) {\tiny 93};
      \node[black] at ( 0.8, -0.4) {\tiny 94};
    \end{tikzpicture}
   \end{tabular}
  }
\label{fig-all}%
\end{figure}
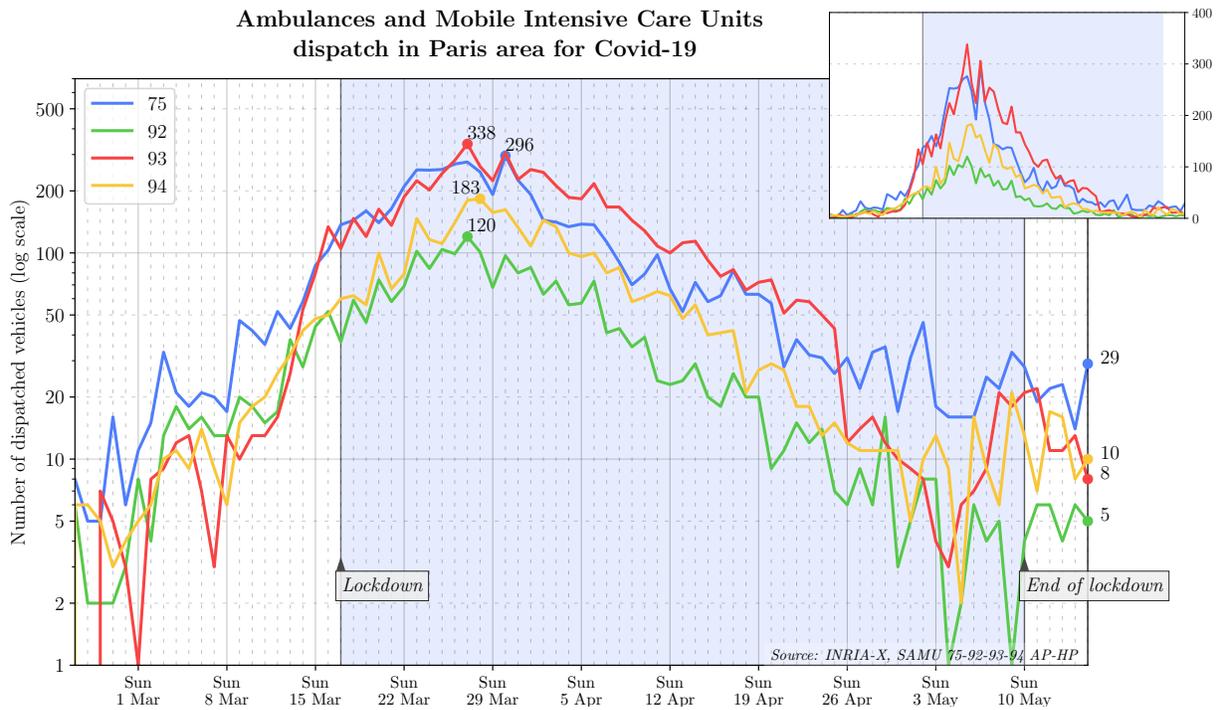

We provide in Table~\ref{table-div} the doubling times of
the number of vehicles dispatched (ambulances and MICU), for the
different departments, measured in days (abbreviation ``d'').

\begin{table}[htbp]
  \begin{center}
    \begin{tabular}{ccccc}
      & Feb. 28\th -- Mar. 15\th  & Mar. 15\th -- Mar. 29\th &  Mar. 29\th -- April 24\th\\\hline %

      75 & 5.9 d         & 9.8  d           &  -9.4  d         \\  %
      92 & 4.9 d         & 10.6 d           &  -8.3  d         \\  %
      93 & 4.2 d         & 8.5  d           &  -10.2 d        \\  %
      94 & 4.6 d         & 6.9  d           &  -7.7  d         \\  %
      \end{tabular}
  \end{center}
    \caption{Doubling time of the number of MICU and ambulances dispatched, for different periods, for each department, obtained by a least squares approximation of the logarithm of this number. The opposite of a negative doubling time yields the halving time.}\label{table-div}
\end{table}

We now draw several conclusions from the previous analysis.

\subsection{The increase in the number of calls for medical advice provides an early, but noisy, indicator
of the epidemic growth}\label{sec-noisy}
As shown in~\Cref{fig-differenttypes},
the peak of the number of calls for medical advice was on March 13\th. However, this date, four days before the lockdown (March 17\th),
is not consistent with epidemiological modeling.
This peak seems rather to be caused by announcements to the population,
see the discussion in \Cref{sec-announce}.

\subsection{The epidemic kinetics vary strongly across neighboring departments}
\label{sec-spatial-varies}
In the initial phase of the epidemic (Feb. 28\th--March 15\th), the doubling time was significantly shorter in the 93 department
(4.2 d) than in central Paris (5.9 d). The 93 department, with 1.6M inhabitants,
is less populated than central Paris (2.1M inhabitants).
Another difference between the departments concerns mobility.
Movement from the population from central Paris to smaller towns and cities
or to countryside were observed,
after March 12\th, the date of the first presidential
address concerning the Covid-19 crisis.

In order to quantify this mobility, we requested information
from Enedis, the company in charge of the electricity distribution
network in France, and also from Orange and SFR, two
operators of mobile phone networks.

Enedis provided us with an estimation of the departure rates
of households, based on a variation of the volume of electricity
consumed, aggregated at the level of departments and districts
(i.e., {\em arrondissements}).

SFR provided us with estimates of daily flows
from the Paris area to other regions, again aggregated at the scale
of the departments or districts, based on mobile phone activity,
confirming this decrease of population.

Orange Flux Vision provided us with daily population estimates,
at the scale of department, based on mobile phone activity.
By March 30\th,
the population, during the night, was estimated
to be 1.6M inhabitants in central Paris, versus
1.35M in the 93. 

However, the epidemic peak was higher in the 93 than in the 75
(338 dispatches versus 296).
The contraction rate in the period after the peak (March 29\th--Apr. 24\th)
was also smaller in the 93, with a halving time of 10.2 days, to be compared
with 9.4 days in the 75. Possible explanations
for these strong spatial discrepancies are
discussed in~\Cref{sec-discrepancies}.

\section{Results -- mathematical modeling}
\subsection{Delay between implementation of sanitary policies and its effect on hospital admissions}\label{sec-delay}%
We explained in~\Cref{sec-piecewise}, based on \Cref{th-1} below,
that the logarithm of an epidemic observable can be approached
by a piecewise linear map with as many pieces as there
are stages of sanitary measures.

So, we look for the best approximation, in the $\ell_1$ norm,
of the logarithm of the number of vehicles dispatched
(ambulances and MICU),
by a piecewise linear map with at most three pieces.
This best approximation is shown on~\Cref{p-phases}.
It is computed by the method of~\Cref{appendix-2}. 

In order to evaluate the influence of a sanitary
measure on the growth of the epidemic,
an approach is to compare the date of the measure
with the date of the change of slope of the logarithmic
curve, consecutive to the measure. This method
is expected to be more robust than,
for instance, a comparison of peak values, because
the best piecewise-linear approximation is obtained
by an optimization procedure {\em taking the whole
sequence into account}. Indeed, a local corruption
of data will not change significantly the date of change of slope,
if the problem is {\em well conditioned}. This is the
case in particular if the difference between consecutive
slopes is sufficiently important. In other words,
we can identify in a more robust manner the time of effect
of a strong measure than of a mild one.

Let us recall the main changes of sanitary measures
in the Paris area, between February and May 2020. We
may distinguish the following phases:

    \begin{itemize}
    \item[-] {\em Initial development of the epidemic}, no general sanitary measures in
      the Paris area, until Feb 29\th, first day of so-called
      ``stade 2'' by the authorities (following ``stade 1'' in which
      measures intended to prevent the introduction of the virus
      in France -- like quarantine in specific cases--  were taken).
    \item[-] {\em ``Stade 2'' (stage 2) measures}:
      general instructions of social
      distancing
      given to the population (e.g., not shaking hands), ban on
      large gatherings.
      Moreover, some large companies created crisis
      committees, and decided to take more
      restrictive measures than the ones required
      by the authorities, including for instance banning meetings with
 more of 10 people, and banning business travels. Restrictive measures
      in companies were deployed gradually during
      the work week from March 2\nd~to March 6\th.
    \item[-] {\em School closure} on March 16\th.
    \item[-] {\em Lockdown} on March 17\th. The lockdown ended
      on May 11\th, throughout the country.
    \end{itemize}

    Hence, we may interpret the variations in the slope
    in the piecewise linear approximation of the logarithm
    of the number of ambulances and MICU dispatched, shown
    on~\Cref{p-phases}, as the effect of sanitary measures.
    The dates where the slope changes are represented in the figure
    by dotted lines. Thus, the latest breakpoint
    of the piecewise linear approximation of the
    75 curve (in blue) arises on March 26\th, to be compared
    with March 30\th~in the 93 (red curve). The dates of 
    breakpoints in the 92 and 94 are intermediate.
    Given the first strong measure (closing of schools)
    was taken on March 16\th, we may evaluate the delay
    between a sanitary measure and its effect on the
    ambulances and MICU dispatch to be between
    10 and 14 days. This corresponds to a delay
    between contamination and occurrence of severe symptoms.

    \begin{figure}[htbp]
      \begin{center}
        \includegraphics[scale=0.55, trim={1cm, 1cm, 2cm, 1cm}, clip]{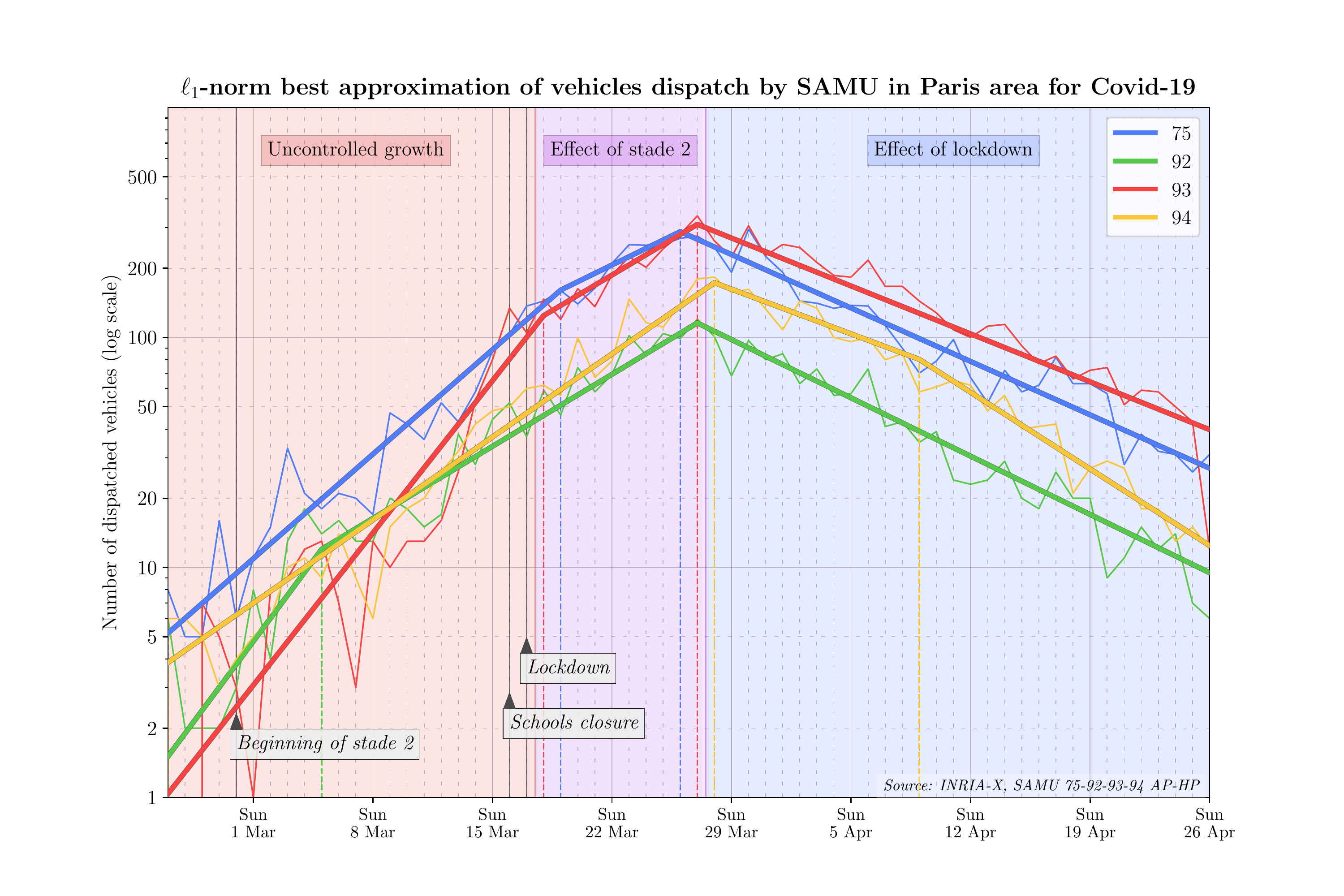}

  \end{center}
  \caption{Logarithm of the number of ambulances dispatched: the effect of the successive sanitary measures}
  \label{p-phases}

\end{figure}

 \subsection{Construction of statistical indicators of epidemic resurgence based on emergency calls}

\begin{figure}[htbp]
  \begin{center}
    \includegraphics[scale=0.56, trim={1.5cm, 1cm, 2cm, 0cm}, clip]{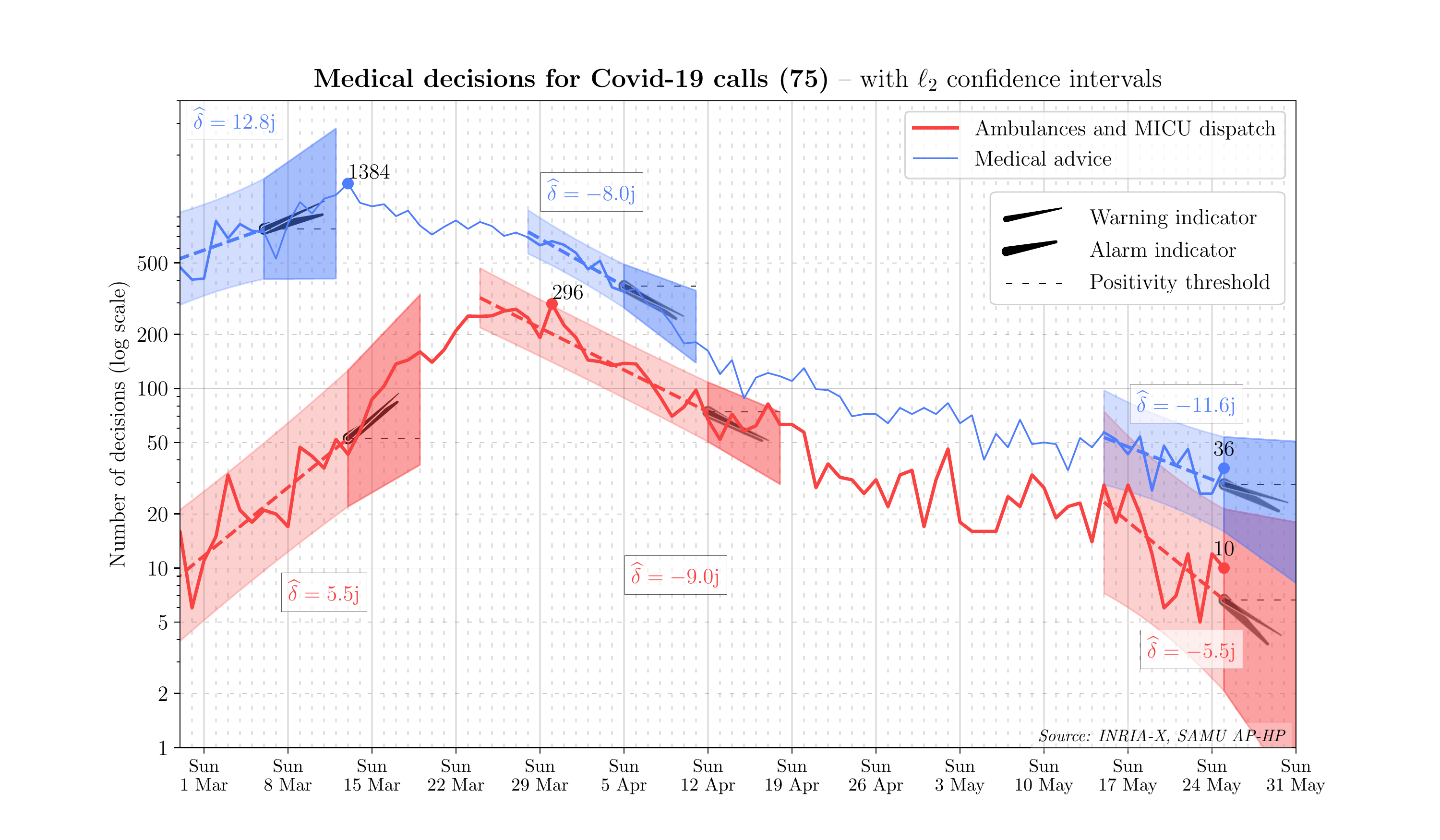}
  \end{center}
  \caption{Short term predictor for SAMU 75, with confidence regions and warning and alarms indicators}
    \label{fig-predictor}
\end{figure}

We implemented the alarm mechanism based on the inference
of doubling times described in~\Cref{ssec-alarm} and further explained in~\Cref{sec-appendixC}. The method
is illustrated on~\Cref{fig-predictor}. Given a time period where the data is known,
we perform a linear regression on the number
of medical advices and vehicles dispatched.

The light shaded, tubular areas around the curves are based on confidence intervals for the fluctuations of the observed log-counts. The dark-shaded, trapezoidal areas prolongate the tubular areas with straight lines, the slopes of which correspond to confidence intervals for the slope of the linear regression. We display such confidence domains for the last known data in May, performing a 6-day forecast based on the last ten days. In order to validate the method, we also display these domains for
older
data in March and April, performing for each a 6-day forecast based on a number of past days. For these time-periods, the short-term confidence domains are seen to satisfactorily contain the data of the
following days.
Observe how the shape of the confidence trapezoids depends on the number of points and the variance of the data used to compute them.

The needle-shaped (or clock hand) indicators depicted in each trapezoidal confidence domain illustrate the alarm mechanism of~\Cref{ssec-alarm}. For each slope inference, there is a $\vartheta_\warning=25\%$ probability that the real-value of the slope we estimated using recent data
is greater than the slope of the thin needle on the picture. Likewise there is a $\vartheta_\alarm=75\%$ probability that the real slope is
greater than the slope of the fat hand. As a result, a warning (resp.\ an alarm) on the dynamics of the medical advice curve should be triggered as soon as the thin needle (resp.\ the fat needle) has a positive angle with respect to the horizontal. We have depicted the horizontal with dashed lines to
enhance readability.

On May 25\nd, no warning nor alarm is triggered, since all the needle-indicators are below the positivity threshold, indicating with at least a 75\% confidence level that based on the last ten days, the two observed signals are on a decreasing trend. Note that due to the relative stagnation of the medical advice curve in the two first weeks of May, computing our indicators a few days earlier (such as May 22\nd) would have raised an warning to arouse vigilance due to the uncertainty on the future trend, but no alarm.

The numbers $\vartheta_\warning$ and $\vartheta_\alarm$ need
to be carefully calibrated, and for this, additional
data for the forecoming weeks may be helpful.

\section{Discussion}\label{sec-discuss}

\subsection{Calls resulting in medical advice are highly influenced by the instructions given to the population}\label{sec-announce}
The blue curve on~\Cref{fig-differenttypes,fig-predictor} %
counts the number of times medical advices was given, of all kinds
(calls resulting in recommendations given to the patient but no vehicle dispatched).
It is generally associated with early events in the unfolding of the pathology, and in particular occurrence of
the first symptoms.
An estimate of 5.1 days between the date of contamination and the date of the first symptoms is given in~\cite{laueretal}, so we may assume that calls for medical advice
are made by patients 5-8 days after contamination.
We observed in~\Cref{sec-noisy} that the peak in the number
of calls resulting in medical advice was on March 13\th. 
Hence, assuming the peak of contaminations was just
before lockdown, the peak of the
curve of new symptom occurrences should occur {\em only several days after
the lockdown time} (March 17\th). This indicates that the curve
of calls resulting in a medical advice did not give
a reliable picture of the epidemic growth around March 13\th.
Indeed, this curve is very sensitive to changes in the instructions
given to the population and to political announcements, notable
examples of which include the following:
recommendation to patients to call emergency
number 15, instead of going directly to emergency departments (to avoid
contamination and overcrowding); -- the presidential announcement on March 12\th~of more restrictive
measures to be deployed from March 16\th, making the population more aware of the
growth of the epidemic.

\subsection{The indicators of medical advice given and ambulances and MICU dispatch can be used to monitor the epidemic}

Setting aside perturbations due to political announcements or changes in the policy for calling SAMU,
the curve  of calls for medical advice should be a reliable and
early estimate of the curve of ambulances dispatched, %
which is triggered at a later stage in the unfolding of the pathology when symptom severity increases.
Thus, it gives an early signal allowing both SAMU  and hospitals to anticipate by several days an increase in load.

We can give a rough estimate
of this delay by considering the peak dates in~\Cref{fig-differenttypes}.
Epidemiological modeling indicates that the number of new contaminations
grows exponentially until a sanitary measure that is strong
enough to contain the epidemic is taken. In the present case, the candidates for
such strong measures are the school closing (March 16\th)
or the lockdown (March 17\th). Considering the last
day previous to the measures, we may assume that
the peak date for new contaminations was
on March 15\th~or March 16\th.
As mentioned above, according to~\cite{laueretal},
the time between contamination and first symptoms
is estimated to be 5.1 days. 
Hence, if calls for medical advice were representative
of first symptoms, the peak value for these calls would have
been between around March 20th or March 21\st.
The peak of the number of dispatches of ambulances and MICU was on March 27\th.
This leads to an estimate of 6-7 days for the
delay between the curve of the true need for medical advice
and the curve of vehicles dispatch.

The alarm mechanism was developed during the crisis, after March 20\th.
Had it been available before, 
in the 75, the early warning would have occurred
by February 24\textsuperscript{th} (for both \covid-19 indicators on medical advice and vehicle dispatch),  alarm would
have been been triggered on February 25\textsuperscript{th} based on medical advice, and a confirmation alarm would have occurred on February 27\textsuperscript{th} based on ambulance and MICU dispatch.

Tracking the same signal at a finer spatial resolution (a neighborhood rather than a department) may enable epidemic surveillance during the period following the lifting of measures such as travel bans. 
Indeed in view of the spatial differentiation in doubling times at the department level, it appears plausible that disparities may also be present at much finer spatial granularities, with resurgences localized to towns or neighborhoods. Deployment of alarm mechanisms constructed from event counts at finer spatial granularities could be used to identify clusters of resurgence early on, and guide subsequent action. 

Seasonal influenza, whose early symptoms can be mistaken for those of \covid-19, can also trigger a linear growth of the logarithms of the number of calls, or vehicle dispatched.
However the slope of this logarithmic curve is expected to be shallower,
owing to the lower contagiosity of influenza. One could thus distinguish, on the basis of the observed slope, whether one is confronted with seasonal flu or with an outbreak of \covid-19.

\subsection{Jumps of the curves of the number of calls may be caused by large clusters or influenced by neighboring countries}
The curves of the number of calls for medical advice, and
of vehicles dispatched (\Cref{fig-differenttypes}), both
jump between February 23\rd\ and 25\th.
Epidemiogical models are unlikely to produce
shocks of this type in the absence of exceptional
factors. Several significant epidemic events occurred at nearby dates,
including:
\begin{enumerate}
\item The development of the Covid-19 epidemic in the north of Italy, a
  region closely tied to France (the first lockdowns occurred around February 21\st\ in the province of Lodi). The school vacation ended on February
  23\th in the Paris area. A significant number of parisians went back
  from Italy during the week-end of February 22\nd-23\th (end of
  school vacation).
 \item A large Evangelist meeting (Semaine de Carême
    de l’Eglise La Porte Ouverte Chrétienne de Bourtzwiller, Haut-Rhin),
    from February 17\th\ to February 21\st,
    identified by Agence Régionale de Santé Grand Est as the
    source of a cluster phenomenon~\cite{arsgrandest}.
  \end{enumerate}
The potential influence of the north Italy epidemic was pointed out by
Paul-Georges Reuter (private communication).  At this stage, the
essential factors are not yet known.
The influence of mobility on the development
of the epidemic in the Paris area
will be studied in a further work.

\subsection{Patients from different areas tend to call the SAMU at different stages of the pathology}
Considering the piecewise linear curves in~\Cref{p-phases},
we note that in the 93, the date of the
    break point is shifted of 3 days, by comparison
    with the 75, suggesting that by this time, the patients
    of 93 were calling Center 15 at a later stage of the evolution
    of the disease. This hypothesis is confirmed by an examination
    of the ratio of the number of MICU dispatched over the number of ambulances
    dispatched. For instance, on March 30\th, there
    were 9 MICU dispatches and 276 other dispatches
    in the 75, to be compared with 25 MICU dispatches
    and 262 other dispatches in the 93, i.e.,
    ratios of 3.3\% in the 75 and 9.5\% in the 93.

    In the same way, the first breakpoints
    of the curves give an indication
    of the times at which ``stade 2'' measures
    influence the epidemic growth.
    These dates range from March 15\th~(for 92)
    to March 22\nd~(for 75). It may be the case
    that these dates are dispersed because
    the change of slope is relatively small, meaning
    that the effect of stade 2 is mild. Indeed, the milder the slope change, the more the estimation of the corresponding date is sensitive to noise. 
    Another effect which may have perturbed the curves
    is the important mobility of the population
    in the 4 departments, between March 12\th~and March 16\th.

\subsection{The strong spatial heterogeneity of the evolution of the epidemic may be explained by local conditions}
\label{sec-discrepancies}\label{sec-plateau}

We observed in~\Cref{sec-spatial-varies} that in the initial phase of the epidemic (Feb. 28\th--March 15\th), the doubling time was significantly shorter in the 93 department, whereas in the contraction phase, the halving time
was significantly higher. 

One may speculate that the contraction rate
in the lockdown phase is influenced by intra-familial contaminations.
In this respect, according to a survey of INSEE,
the national institute of statistics, the average size of a household
is of 2.6 in the 93, versus 1.9 in the 75 (values in 2016~\cite{Insee}).

We also remark that just after the peak on March 27,
and up to April 6, the curve of the department
93 on \Cref{fig-all} has the shape
of a high, oscillating plateau, 
decaying more slowly than the curve of the department
75. This may be caused by changes in the nature of the dominant mode
of contaminations, intra-familial contamination becoming an essential
part of the kinetics during lockdown.

One may also speculate that the blowup rate
in the initial phase is higher when the population
is more dependent on public transport, or
working in jobs with more contamination risk.
These aspects will
be further studied elsewhere.

After Stade 2 was announced, during
the period from March 2\nd~to March 6\th, a number
of large companies  took specific measures (e.g., forbidding
avoidable small group meetings,
enforcing travel restrictions, restricting office access),
in addition to the general measures (not shaking hands, forbidding large meetings)
enforced by the authorities. This may have led to a decrease
of the number of contamination on the workplace, and
one explanation for the increase of the doubling time.

\section{Epidemiological model based on transport PDE}
\label{sec-epidemio}
\label{sec-pde}
\subsection{Taking delays into account: a transport PDE SEIR model}
\label{sec-epidemio-1}
We now introduce a multi-compartment transport PDE model,
representing the dynamics of Covid-19.
As explained in~\Cref{sec-model}, in contrast
to ODE models, that assume that the transition time from a compartment to the next one has an exponential distribution, PDE models capture
{\em transition delays} bounded away from zero,
an essential feature of Covid-19. An interest of this PDE
model also lies in its unifying character:
it includes as special cases, or as variations,
SEIR ODE models that have been considered~\cite{crepey,colizza}.

We shall keep the traditional decomposition of individuals in compartments,
``susceptible'' (S), ``exposed'' (E), ``infectious'' (I), and ``removed''
from the contamination chain (R), as explained in~\Cref{sec-model}, but the state
variables attached to the $E$ and $I$ compartments will
take the time elapsed in the compartment into account,
and thus, will be infinite dimensional.

For all $t\geq 0$, we denote by $n_E(x,t)$ the density
of the number of individuals that were contaminated $x$ time units
before time $t$, and that are not yet infectious at time $t$, i.e.,
the number of exposed individuals that began to be exposed
at time $t-x$. 
Then, the size of the exposed population at time $t$ is given
by
\begin{align}
E(t) &= \int_0^\infty  n_E(x,t)\,\mathrm{d}x \enspace .\label{pop-exposed}
\end{align}
Similarly, we denote by $n_I(x,t)$ the density of the number of individuals
that became infectious $x$ time units before time $t$,
and that are not yet removed from the contamination chain
at that time. Then, the size of the infectious population
at time $t$ is given by
\begin{align}
  I(t) &=  \int_0^\infty  n_I(x,t)\,\mathrm{d}x \enspace .
  \label{pop-infectious}
\end{align}
Finally, we denote by $S(t)$ the number of susceptible individuals at time $t$,
and by $R(t)$ the number of individuals that have been removed from the contamination chain before time $t$. 

The total population at time $t$ is given by
\[ N(t):= S(t) + E(t) + I(t) + R(t) \enspace .
\]
We consider the following system of PDE and ODE, with
integral terms in the boundary conditions:
\begin{subequations}
\begin{align}
&  \frac{\mathrm{d}S}{\mathrm{d}t} = -   \frac{S(t)}{N(t)}\int_0^\infty K_{I\to E}(x,t) {n_I(x,t)} \,\mathrm{d}x \enspace, \label{e-leave-S}\\
  &
    n_E(0,t) =  \frac{S(t)}{N(t)}\int_0^\infty K_{I\to E}(x,t) {n_I(x,t)}\,\mathrm{d}x \enspace,\qquad 
  \frac{\partial n_E}{\partial t} (x,t)+ \frac{\partial n_E}{\partial x}(x,t) + K_{E\to I}(x,t) n_E(x,t)  = 0 \enspace, \label{e-Lotka}\\
& n_I(0,t) = \int_0^\infty K_{E\to I}(x,t) n_E(x,t)\,\mathrm{d}x \enspace,\qquad 
\frac{\partial n_I}{\partial t}(x,t) + \frac{\partial n_I}{\partial x}(x,t) + K_{I\to R}(x,t) n_I(x,t)  = 0 \enspace,\\
& \frac{\mathrm{d}R}{\mathrm{d}t} = \int_0^\infty K_{I\to R}(x,t) n_I(x,t)\,\mathrm{d}x \enspace.\label{e-removed}
\end{align}\label{e-transport}
\end{subequations}
We assume that an initial condition at time $0$, $S(0)$, $n_E(\cdot,0)$, $n_I(\cdot,0)$ and $R(0)$ is given.

This is inspired by the so called ``age structured models''
considered in population dynamics.
Kermack and McKendrick developed the first model
of this kind to analyze the Plague
epidemy of Dec.~1905 -- July~1906 in Mumbai~\cite{kermack-mckendrick}. Von Forster~\cite{vonforster} studied a similar model.
Nowadays, these models are used as a general tool in population
dynamics, with applications to biology and ecology~\cite{webb,PerthameBook,MMP},

In these models, ``age'' refers to the age elapsed in a compartment --
each transition to a new compartment resets to zero
the ``age'' of an individual. In contrast, in the classical SEIR 
literature based on ODE, the standard
notion of age (time elapsed since birth) is taken into account,
via a contact matrix tabulating age-dependent infectiosity
rates~\cite{crepey}. These two notions of age should not be confused.
In the sequel, we shall use quotes, as in ``age'', to denote
the age in a compartment, and will omit quotes to
denote the ordinary age (since birth).

We suppose that $K_{I\to E}$, $K_{E\to I}$ and $K_{I\to R}$ are given
{\em nonnegative} functions. The value $K_{E\to I}(x,t)$ gives
the departure rate from the compartment $E$ to the compartment
$I$, for individuals of ``age'' $x$ in the compartment $E$,
at time $t$.
Similarly, $K_{I\to R}(x,t)$ gives
the departure rate from the compartment $I$ to the compartment
$R$. As in the classical SEIR model, the departure
term from the susceptible compartment, i.e., the right-hand-side
of~\eqref{e-leave-S}
is bilinear in the number $S(t)$ of susceptible
individuals and in the population of infectious
individuals $n_I(\cdot,t)$, and we normalize by the size
of the population $N(t)$.  The term $K_{I\to E}(x,t)$ can be
interpreted as an infection rate.

Differentiating $N(t)$ with respect to time, using the system above,
and assuming that for all $t\geq 0$,
$n_E(x,t)$ and $n_I(x,t)$ vanish when $x$ tends to infinity,
we verify that the total population $N(t)$ is independent of time.

When the functions $K_{I\to E}, K_{E\to I}$ and $K_{I\to R}$
are constant, taking into account~\eqref{pop-exposed} and~\eqref{pop-infectious}, we recover the classical SEIR model from the dynamics~\eqref{e-transport}:
\begin{subequations}
\begin{align}
  &  \dot{S} = -   \frac{S}{N} K_{I\to E} I
  \enspace, \label{e-leave-S-SEIR}\\
  &
\dot{E} =  \frac{S}{N} K_{I\to E} I - K_{E\to I} E \enspace, \label{e-Lotka-SEIR}\\
&
\dot{I}= K_{E\to I} E - K_{I\to R} I \enspace, \label{e-Idif}\\ 
& \dot{R} = K_{I\to R} I  \enspace.
\end{align}\label{e-seir}
\end{subequations}

In the sequel, we shall consider~\eqref{e-transport}
instead of~\eqref{e-seir}, and we shall assume that
the rates $K_{E\to I}(x,t)=K_{E\to I}(x)$ and $K_{I\to R}(x,t)=K_{I\to R}(x)$
are functions of $x$, independent of time.
The rate $K_{I\to E}$ will have the product form
\[
K_{I\to E}(x,t) =\control(t) \psi(x) \enspace.
\]
The function $\psi(\cdot)$ is fixed, it is nonnegative
and not {\em a.e.}~zero.
In this way, the infectiosity of an individual depends on his
``age'' in the infectious phase, whereas the term $\control(t)$ represents
the control of the epidemic by sanitary measures (social distancing,
wearing masks, closing schools, lockdown, etc.). We shall
assume that the infectiosity rate
$K_{I\to E}(x,t)$ is the only parameter which can be controlled,
hence, $\control(\cdot)$ is a decision variable.
A variant of the ODE model~\eqref{e-seir},
in which $K_{I\to E}$ depends on time,
but not on $x$, is considered in~\cite{CS-Chang2020}.
Other versions, including a modification of the SEIR model
leading to a time delay differential equation,
are discussed in~\cite{magaldelay}.

For epidemics in their early stages, i.e.,
when the number of individuals in the exposed, infectious,
or removed compartments is negligible with respect to the number
of susceptible individuals, the classical SEIR model is well-approximated
by a linear system (see e.g.~\cite{kermack-mckendrick,bacaer}) tracking only the populations in the (E) and (I) compartments.
As noted in~\Cref{sec-model}, the fraction  of the French population exposed prior to May 11 is estimated of~5.7\% (see~\cite{salje}),
which justifies reliance on this linear approximation in our context. The same approximation
applies to the present PDE model. This is translated
to the assumption $S(t)/N(t)\simeq 1$, and we are reduced to the following system:
\begin{subequations}
\begin{align}
  &
    n_E(0,t) =  \int_0^\infty \!\!\control(t) \psi(x) {n_I(x,t)}\,\mathrm{d}x \enspace,\qquad 
  \frac{\partial n_E}{\partial t} (x,t)+ \frac{\partial n_E}{\partial x}(x,t) + K_{E\to I}(x) n_E(x,t)  = 0 \,, \quad \text{for }x>0 \,,\label{e-Lotkastat}\\
& n_I(0,t) = \int_0^\infty \!\!K_{E\to I}(x) n_E(x,t)\,\mathrm{d}x \,,\qquad 
  \frac{\partial n_I}{\partial t}(x,t) + \frac{\partial n_I}{\partial x}(x,t) + K_{I\to R}(x) n_I(x,t)  = 0 \enspace, \quad \text{for } x>0 \enspace.\label{e-leave}
\end{align}\label{e-transportsimple}
\end{subequations}
This is a two-compartment generalization of the renewal equation,
studied in Chapter~3 of~\cite{PerthameBook}.

In the sequel, we shall assume
that there is a maximal ``age'' $x^*_E$ of an individual
in the exposed state. Similarly, we shall assume that there
is a maximal ``age'' $x^*_I$ of an individual in the infectious
state. These assumptions, which are consistent
with epidemiological observations~\cite{laueretal},
will be incorporated in our model by forcing
all remaining exposed individuals of ``age'' $x_E^*$
to become infectious, with ``age'' $0$.
Similarly, all the remaining infectious individuals
are removed when reaching ``age'' $x_I^*$.
So, the function $n_E$ is now only defined
on the interval $[0,x_E^*]$, and similarly,
$n_I$ is only defined on $[0,x_I^*]$.
This leads to the following system: 
\begin{subequations}
\begin{align}
  &\!\!\!
  n_E(0,t) =  \int_0^{x_I^*}\!\!\! \control(t) \psi(x) {n_I(x,t)}\,\mathrm{d}x \,,
  \quad 
  \frac{\partial n_E}{\partial t} (x,t)+ \frac{\partial n_E}{\partial x}(x,t) + K_{E\to I}(x) n_E(x,t)  = 0 \,,\quad \text{for }0<x<x_E^* \,,\label{e-Lotkastat-compact}\\
  &\!\!\! n_I(0,t) =  \int_0^{x_E^*} K_{E\to I}(x) n_E(x,t)\,\mathrm{d}x + n_E(x^*_E,t) \,,\label{e-bd2}\\
  & %
  \qquad  \qquad  \qquad  \qquad  \qquad  \qquad  \qquad  \quad
  \frac{\partial n_I}{\partial t}(x,t) + \frac{\partial n_I}{\partial x}(x,t) + K_{I\to R}(x) n_I(x,t)  = 0 \,, \quad \text{for }0<x<x_I^*\,,\label{e-leavec-compact}\\
     &\!\!\!\frac{\dd R}{\dd t}(t) = \int_0^{x_I^*} K_{I\to R}(x) n_I(x,t)\mathrm{d} x  + n_I(x_I^*)  \enspace.
  \label{e-seir-compact}
\end{align}\label{e-transportsimple-compact}
\end{subequations}
This system may be obtained as a specialization of~\eqref{e-transportsimple},
in which $K_{E\to I}(x) $ is replaced by
$K_{E\to I}(x) \unit_{[0,x_E^*]}(x)+ \delta_{x_E^*}(x)$,
where $\unit$ denotes the indicator function
of a set, and $\delta$ denotes Dirac's delta function.

Note that the above model is still relevant when $x^*_E=0$.
Then, the partial
differential equation in \eqref{e-Lotkastat-compact} disappears, and
we are left with a PDE model modelling an infinite dimensional
compartement of infectious
individuals, without an explicit ``exposed but not yet infectious''
compartment.
This is similar to the original model of~\cite{kermack-mckendrick}.
In contrast, the maximal time elapsed by an individual in the infectious state, $x^*_I$, must be positive. Otherwise, the integral term in~\eqref{e-Lotkastat-compact}, representing contaminations, vanishes.

We shall assume, in the sequel, that the following assumption
holds.
\begin{assumption}\label{as-nonden}
  The functions $K_{E\to I}(\cdot)$,
  defined on $[0,x_E^*]$, and $\psi(\cdot)$ and $K_{I\to R}(\cdot)$,
  defined on $[0,x_I^*]$, are nonnegative, measurable and bounded.
Moreover, the function $\psi$ does not vanish {\em a.e.}\ and the point $x_I^*>0$ is the maximum of the essential support of
the function $\psi$. 
\end{assumption}
Indeed, considering the boundary condition in~\eqref{e-Lotkastat-compact},
we see that a population of ``age'' $x> \max\operatorname{ess\, supp}\psi$
in the infected ($I$) compartment 
will not participate any more to the contamination chain. Hence,
the last part of Assumption~\ref{as-nonden} is needed to interpret $R$ has the number of {\em all} the removed individuals.

Systems of PDE of this nature have been studied
in particular by Michel, Mischler and Perthame,
see~\cite{MMP,PerthameBook}, and also,
with an abstract semigroup perspective,
in the work by Mischler and Scher~\cite{MischlerScher}.

Then, using the boundedness of the coefficients (Assumption~\ref{as-nonden}),
and arguing as in the proof
of Theorem~3.1 of~\cite{PerthameBook} -- which concerns
the case of a single compartment -- one can show
that the system~\eqref{e-transportsimple-compact}
admits a unique solution in the distribution
sense $n:=(n_E,n_I)$ with
$n_E \in \mathcal{C}(\R_{\geq 0}, L^1([0,x_E^*]))$
and $n_I\in \mathcal{C}(\R_{\geq 0}, L^1([0,x_I^*]))$.
Hence, we can associate to the PDE~\eqref{e-transportsimple} a
well defined family of time evolution linear operators $(T_{s,t})_{t\geq s\geq 0}$,
acting
on the space $L^1([0,x^*_E])\times L^1([0,x^*_I])$.
The operator $T_{s,t}$ maps an initial
condition at time $s\geq 0$, that is a couple of functions
$n(\cdot,s):=(n_E(\cdot,s),n_I(\cdot,s))$, %
to the couple of functions $n(\cdot,t):=(n_E(\cdot,t),n_I(\cdot,t))$
at $t\geq s$. These operators are order preserving,
meaning that, if $n^1(\cdot,s)$ and $n^2(\cdot,s)$ are two initial
conditions such that $n^1_E(x,s)\leq n^2_E(x,s)$
and $n^1_I(x,s)\leq n^2_I(x,s)$
for all $x\geq 0$,
then the inequalities $n^1_E(x,t)\leq n^2_E(x,t)$
and $n^1_I(x,t)\leq n^2_I(x,t)$ hold for
all $x\geq 0$ and for all $t\geq s$.

An alternative modeling, more in the spirit of~\cite{kermack-mckendrick},
would be to consider
a single compartment, describing
the evolution of the density $n(x,t)$
of individuals that were contaminated at time $t-x$ by
the system:
\begin{align}
  &
  n(0,t) =  \int_0^{\infty}\!\! \control(t) \psi(x) {n(x,t)}\,\mathrm{d}x \,,
  \quad 
  \frac{\partial n}{\partial t} (x,t)+ \frac{\partial n}{\partial x}(x,t) + K(x) n(x,t)  = 0 \,,\quad \text{for }0<x<x^* \,,\label{e-Lotkastat-compactnew}
\end{align}
where $x^*>x^*_E$ is fixed, and $\psi(x)=0$ for $x<x^*_E$.
Then, $E(t) = \int_0^{x_E^*} n(x,t) \mathrm{d}x $ yields the size
of the exposed compartment. However, we prefer
the model~\eqref{e-transportsimple-compact} as it allows us to represent
variable incubation times.

The system~\eqref{e-seir-compact} can be extended to
represent infectiosity rates that depends on
the ages (time elapsed since birth) of individuals,
with infectiosity rates given by a contact matrix,
as in~\cite{crepey}. It suffices to split
each compartment in sub-compartments, corresponding
to different age groups. This will be detailed
in a further work.

\subsection{A Perron-Frobenius Eigenproblem for Transport PDE}
\label{sec-PFT}

When the control $\control(t)$ is constant and positive,
the family of time evolution operators $(T_{s,t})_{t\geq s\geq 0}$ is
determined by the semigroup $(S_{t}=T_{0,t})_{t\geq 0}$, and 
the long term evolution of the dynamical system~\eqref{e-transport}
can be studied by means of the {\em Perron--Frobenius eigenproblem}
\begin{subequations}
\begin{align}
  &
  \bar n_E(0) =  \int_0^{x_I^*} \control \psi(x) {\bar n_I(x)}\,\mathrm{d}x \enspace,\quad 
 \frac{\mathrm{d} \bar n_E}{\mathrm{d} x}(x) + (\lambda + K_{E\to I}(x)) \bar n_E(x)  = 0 \, \quad\text{for }0<x<x_E^*\, ,\label{e-Lotkaeig}\\
& \bar n_I(0) =  \int_0^{x_E^*} K_{E\to I}(x) \bar n_E(x)\,\mathrm{d}x + n_E(x_E^*)\enspace,\quad 
 \frac{\mathrm d \bar n_I}{\mathrm d x}(x) +
 ( \lambda + K_{I\to R}(x) )\bar n_I(x)  = 0 \,\quad\text{for }0<x<x_I^* \,,
 \label{e-eig2}
\end{align}\label{e-transportsimpleeig}
\end{subequations}
where $\bar n:=( \bar n_I(\cdot),\bar n_E(\cdot))$ is a nonnegative eigenvector,
and $\lambda$ is the eigenvalue.

We make a general observation from Perron-Frobenius theory.
\begin{lemma}\label{prop-unique}
  Let $w=(w_E,w_I)$, with $w_E\in L^1([0,x_E^*])$
  and $w_I\in L^1([0,x_I^*])$, be such that
  \begin{align}
\alpha   \bar{n} \leq w \leq \beta \bar{n}\label{e-sandwitch}
\end{align}
for some $\alpha,\beta>0$. Then,
\begin{align}
\alpha \exp(\lambda t) \bar n \leq 
S_t w \leq \beta \exp(\lambda t) \bar n ,\qquad \text{for all } t\geq 0 \enspace.
\label{e-asymp}
\end{align}
\end{lemma}
\begin{proof}
  This follows from the order preserving and linear
  character of the semigroup $S_t$, together
  with $S_t\bar{n} =\exp(\lambda t)\bar n$.
  \end{proof}
This shows that the eigenvalue $\lambda$ determines the growth rate of $n(t,x)$
as $t\to \infty$, under the assumption that the initial population $w:=n(0,\cdot)$
is comparable with the eigenvector $\bar{n}$, meaning that inequality~\eqref{e-sandwitch}
below holds for some $\alpha,\beta>0$.

Since the functions $K_{E\to I}$ and $K_{I\to R}$
are independent of time, the existence of
a positive eigenvector is an elementary result:

\begin{proposition}\label{prop-eig}
Suppose  $\mu>0$ and that Assumption~\ref{as-nonden} holds.
Then, the eigenproblem~\eqref{e-transportsimpleeig}
has a solution
$(\bar{n},\lambda)$, where
$\bar{n}=(\bar{n}_E,\bar{n}_I)$,
the functions $\bar{n}_E$ and $\bar{n}_I$
are continuous and positive, and $\lambda\in \R$.
Moreover, the eigenvalue $\lambda$
is unique, and the eigenvector $\bar{n}$ satisfying
the latter conditions is unique
up to a multiplicative constant.
\end{proposition}

The proof of this proposition exploits a classical
argument in renewal theory, see Lemma~3.1 p.~57 of~\cite{PerthameBook}.
We give the proof, leading to a semi-explicit representation
of the eigenvector, which we shall need in~\Cref{sec-tropical}.

\begin{proof}
We next provide a semi-explicit formula for the eigenvector.
We set
  \begin{equation}\label{defFlambda}
  F^\lambda_{E\to I}(x):= \int_0^x (\lambda + K_{E\to I}(z))\mathrm{d}z
  \,,\text{ for }0\leq x\leq x_E^*\, \quad
  F^\lambda_{I\to R}(x):= \int_0^x (\lambda + K_{I\to R}(z))\mathrm{d}z
  \,,\text{ for }0\leq x\leq x_I^*\, .
  \end{equation}
  Let $\bar n=(\bar{n}_E,\bar{n}_I)$ be an eigenvector associated to the eigenvalue $\lambda$, so that it satisfies~\eqref{e-transportsimpleeig},
and assume that $\bar n_E(\cdot)$ is continuous on $(0,x_E^*)$, and $\bar n_I(\cdot)$ is continuous on $[0,x_I^*]$.
Integrating the differential equations in~\eqref{e-transportsimpleeig},
and using the continuity of $\bar n$ and of the following expressions, 
  we see that $\bar n$ necessarily satisfies
  \begin{align}
    \bar n_E(x) &=  \exp\big( - F^\lambda_{E\to I}(x)\big) \bar n_E(0)\,,
    \text{ for } 0\leq  x \leq x_E^*\,
    \quad
    \bar n_I(y)=  \exp\big( - F^\lambda_{I\to R}(y)\big) \bar n_I(0)\,,
    \text{ for } 0 \leq  y\leq x_I^*\,.
    \label{e-explicit}
  \end{align}
  Conversely, if $\bar n$ satisfies the above expressions, then
it satisfies the differential equations in~\eqref{e-transportsimpleeig} and is
continuous.
  Using~\eqref{e-explicit}, together with the boundary conditions in~\eqref{e-transportsimpleeig}, and the assumption that $  \psi$ does not vanish {\em a.e.},
  and that $\mu>0$,
we deduce that if $\bar n_E(x)> 0$ for some $x\in [0,x_E^*]$ or
$\bar n_I(y)> 0$ for some $y\in [0,x_I^*]$, then
 $\bar n_E(x)> 0$ for all $x\in [0,x_E^*]$ and
$\bar n_I(y)> 0$ for all $y\in [0,x_I^*]$.
Then, if a continuous nonnegative eigenvector $\bar n$ exists, it is
everywhere positive.
Moreover,
\eqref{e-explicit} and the boundary condition in~\eqref{e-eig2} entail
that the eigenvector $\bar n$ is unique, up to a scalar multiple.

Using also the boundary condition in~\eqref{e-Lotkaeig}, %
and specializing~\eqref{e-explicit} to $x=x_E^*$ and $y=x_I^*$,
we deduce that $\control G^\lambda \bar{n}_E(0)= \bar{n}_E(0)$,
  where 
\[  G^\lambda = \Big(\int_0^{x_I^*}\psi(x)
  \exp\big( - F_{I\to R}^\lambda (x) \big)\dd x\Big)
 \Big(
    \int_0^{x_E^*} K_{E\to I}(y)
    \exp\big( - F^\lambda_{E\to I} (y)\big)\mathrm{d}y
    +
    \exp\big( - F^\lambda_{E\to I}(x_E^*)\big)
    \Big)\,.
    \]
    Therefore, for an eigenvector $\bar n$ to exist, we must solve
    the equation $\control G^\lambda=1$ (the so-called
    ``characteristic equation'' in renewal theory).
    Since the functions $K_{E\to I}$, $K_{I\to R}$ and $\psi$ are nonnegative and
    integrable,
  and $\psi$ is nonzero on a set of positive measure,
  we deduce that $\lim_{\lambda \to-\infty} G^{\lambda}=+\infty$.
  We also have $\lim_{\lambda \to +\infty}G^{\lambda}=0$. Moreover,
  the map $\lambda \mapsto G^\lambda$ is continuous.
  Since $\control>0$, by the intermediate value theorem,
  we can find $\lambda$ such that $\control G^\lambda =1$,
  and this $\lambda$ is the eigenvalue.

  We showed that any nonnegative continuous eigenvector is positive.

  Hence, any two nonnegative eigenvectors $\bar{n}^1$ and $\bar{n}^2$ with
  eigenvalues $\lambda_1$ and $\lambda_2$ satisfy
  $\alpha \bar{n}^1\leq \bar{n}^2\leq \beta \bar{n}^1$ for some $\alpha,\beta>0$,
  and it follows from~\Cref{prop-unique} that $\alpha \exp(\lambda_1 t)
  \bar{n}^1\leq \exp(\lambda_2 t) \bar{n}^2\leq \beta\exp(\lambda_2 t) \bar{n}^1$,
  which entails that $\lambda_1=\lambda_2$, showing that the eigenvalue
  associated with a nonegative eigenvector is unique. Alternatively,
  the uniqueness of this eigenvalue follows
  from the strictly decreasing character of the map
  $\lambda \mapsto G^\lambda$.
  \end{proof}

The asymptotic
bound~\eqref{e-asymp} can be reinforced,
by showing that, for all positive initial conditions $w$,
\begin{align}
  S_tw  = C_1(w) \bar{n} \exp(\lambda t)
  + O(\exp(\lambda_2 t)) \enspace , \qquad \text{as } t\to\infty \enspace,
  \label{e-specgap}
\end{align}
for some positive constant $C_1(w)$,
and $\lambda_2<\lambda$. This result,
with an explicit control of $\lambda_2$,
can be obtained as follows. We make
a diagonal scaling, using the positive eigenvector,
and we normalize the semigroup to make the Perron eigenvalue
$\lambda$ equal to zero. This leads to the semigroup
\[
(\tilde{S}_t w)(x) := \exp(-\lambda t)\bar{n}^{-1}(x) [S_t ( w\bar{n})](x)
\, .
\]
In potential theory, a version of
this scaling is known as {\em Doob's h-transform} (see e.g.~\cite{dynkin}). 
The semigroup $\tilde{S}_t$ obtained in this way
is associated with a Markov process, and, so, the spectral
gap of this semigroup can be bounded in terms of Doeblin's ergodicity
coefficient~\cite{gaubertquIEuli2013,bansaye}, leading
to~\eqref{e-specgap}.
These aspects will be detailed elsewhere.
Alternatively, the
relative entropy inequality technique of~\cite{MMP}
allows one to establish the convergence
of $n(\cdot,t)$ to the eigenvector, modulo multiplicative
constants, as $t$ tends to infinity. 

\subsection{Universality of the log-rate of epidemic observables}\label{subsec-observables}
Epidemic observables are obtained by applying a continuous linear form
to the state variable.
Supposing
that $n_I(\cdot,t)$ is a continuous function, an
epidemic observable will be of the form
\begin{align}
  Y_\kappa(t) =\varphi(n(\cdot,t)):=
\int_0^{x_I^*} n_I(x,t)\,\mathrm{d}\kappa(x)
\label{e-delayed}
\enspace ,\end{align}
where $\dd\kappa(x)$ is a nonnegative nonzero
Borel measure.  Epidemic events anterior
to the infectious phase, like contamination,
are by nature hard to detect, so the observable
depends only on $n_I$.

\begin{proposition}\label{prop-universal}
  Suppose that Assumption~\ref{as-nonden} holds,
  let $(\lambda,\bar{n})$ denote the solution
  of the Perron-Frobenius eigenproblem~\eqref{e-transportsimpleeig},
  and suppose that for some $T>0$, there exist positive
  constants $\alpha,\beta$ such that $\alpha \bar n\leq n(\cdot,T)\leq \beta \bar n$.
  Then, for all epidemic observables of the
 form~\eqref{e-delayed}, the map $t\mapsto \log Y_\kappa(t)- \lambda t$
 is bounded. A fortiori,
  \[
  \lim_{t\to\infty} \frac{1}{t} \log Y_\kappa(t) = \lambda \enspace .
  \]
\end{proposition}
\begin{proof}
  Taking $w=n(\cdot,T)$ in~\Cref{prop-unique},
  we deduce that
  \[
\alpha \exp(\lambda t) \bar{n}\leq 
S_{T+t} n(\cdot,0)=n(\cdot,T+t)
\leq \beta \exp(\lambda t) \bar{n}\,, \text{ for }t>0 \enspace .
\]
It follows that $\log \alpha + \lambda t + \log \varphi(\bar{n})
\leq \log Y_\kappa(T+t)\leq \log \beta + \lambda t + \log \varphi(\bar{n})$.
\end{proof}

A simple example of observable, discussed in~\Cref{sec-model},
consists of {\em pure delays}. For instance,
we assumed that the number of dispatches
of MICU is given by $Y_{\micu}(t) = \pi_{\micu} C(t-\tau_{\micu})$
where $C(t)$ is the number of contaminations
at time $t$, $\pi_{\micu}$ the proportion of contaminated
patients who will need a MICU transport, and
$\tau_{\micu}$ a fixed delay. This can be obtained
as a special case of~\eqref{e-delayed},
taking $K_{E\to I}\equiv 0$, so that
the transition tom $E$ to $I$ occurs
always at time $x_E^*$, and $\dd \kappa:=\pi_{\micu}\delta_{\tau_{\micu}-x_E^*}$,
where $\delta$ is the Dirac $\delta$ function.

Other events can be considered: medical advice,
EMT dispatch, admission to ICU, or decease. These
events corresponds to different values of the proportion $\pi$ and
of the delay $\tau$.
By \Cref{prop-universal}, the rate $\lim_{t\to\infty} t^{-1}\log Y(t)$
will be the same for all the corresponding observables, although the convergence
of the function $t^{-1}\log Y(t)$ to its limit
will be observed in a delayed manner,
for observables corresponding to the latest stages of the
pathology.

\subsection{Discrete versions of the epidemiological model}
The reader interested in ODE model of epidemics
might wish to note that
the previous analysis applies
to such finite dimensional models.
Instead of the transport PDE~\eqref{e-transportsimple}, 
we may consider an ODE of the form
\begin{align}
\dot{v} = Mv \label{e-metzler}
\end{align}
where $v(t)\in \R^n$ and $M$ is a $n\times n$ matrix with non-negative off-diagonal terms, a so-called {\em Metzler matrix}. In the original SEIR model~\cite{bacaer}, the matrix $M$, obtained by considering
the $(E,I)$-block equations~\eqref{e-Lotka-SEIR}, \eqref{e-Idif},
with $S/N\simeq 1$, is of dimension $2$. In the generalizations
of the SEIR model considered in~\cite{crepey,colizza}, the dimension
$n$ is increased to account for other compartments. 
One can also discretize the PDE system~\eqref{e-transportsimple} using
a monotone (upwind) finite difference scheme,
and this leads to a system of the form~\eqref{e-metzler}.

In all these finite dimensional models,
the matrix $M$ is Metzler and irreducible. 
Then, the Perron--Frobenius theorem for linear, order-preserving semigroups (see~\cite{berman}) implies that $M$ admits a unique eigenvalue $\lambda$ of maximal real part. Furthermore $\lambda$ is algebraically
simple and real, and its associated eigenvector $u$ has strictly positive coordinates.
Then, it follows from the spectral theorem that
\[ v(t)
= \exp(\lambda t)u  + o(\exp(\lambda_2 t)) 
\]
as $t\to \infty$, where $\lambda_2$ is the maximal real
part of an eigenvalue of $M$ distinct from $\lambda$.
Again, in this discrete model, an epidemic observable $Y(t)$ is obtained
by applying a nonnegative linear form to the vector $v(t)$, i.e,
$Y(t)=\varphi^\top v(t)$, for some nonnegative column vector
$\varphi$.

\section{Tropicalization of the logarithm of nonnegative observables of switched Perron--Frobenius dynamics}\label{sec-tropical}
\subsection{Hilbert's geometry applied to piecewise linear approximation}
We introduce an abstract setting, which captures
epidemiological models in which most individuals
are susceptible.
This setting applies, in particular, to the transport PDE
model of~\eqref{e-transportsimple}, when the transition functions are supported
by compact intervals, and to the general finite dimensional
Metzler model~\eqref{e-metzler}.

We consider $(V,\leq)$, a partially ordered
Banach space,
with topological dual $V'$. We denote by $V_{\geq 0}:= \{v\in V\mid v\geq 0\}$
the set of nonnegative elements of $V$, which is a convex
cone.
This cone must be pointed (i.e., $V_{\geq 0} \cap (-V_{\geq 0})=\{0\}$),
since the relation $\leq$ is a partial order.
We require this cone to be closed.

We consider a sequence of $m$ semigroups
$S^i=(S^i_t)_{t\geq 0}$, for $i\in [m]$, where $[m]:=\{1,\dots,m\}$.
We assume that
for all $i\in [m]$, and for all $t\geq 0$,
$S^i_t$ is a bounded linear operator from $V$
to itself, and that the semigroup property holds,
i.e., $S^i_{t+s}= S^i_t\circ S^i_s$. We shall say that the
semigroup $S^i$ is {\em order preserving} if, for all $v\in V_{\geq 0}$,
and for all $t\geq 0$, $S^i_t v\in V_{\geq 0}$.

We shall consider commutation instants, $t_0:=0<t_1<\dots< t_{m-1}$,
These instants will correspond to significant epidemiological dates,
for instance, dates at which sanitary measures are taken. We
set $t_m:=+\infty$.

We select an initial condition $v_0 \in V_{\geq 0}$, and consider
the abstract dynamical system obtained by switching between the
evolutions determined by the semigroups
$S^1,\dots,S^m$, at the successive times $t_1,\dots, t_{m-1}$.
The state of this dynamical system, at time $t\in [t_j, t_{j+1})$,
is given by
\begin{equation}\label{evolPerron}
v_t := S^{j+1}_{t-t_j} \circ S^{j}_{t_j -t_{j-1}}\circ \dots \circ S^1_{t_1-t_0}(v_0)
\enspace .
\end{equation}

Recall that
a {\em part} of the closed convex cone $V_{\geq 0}$ is an equivalence class
for the relation $\sim$ such that, for $v,w\in V_{\geq 0}$,
we have $v\sim w$ if and only if there exists two positive
constants $\alpha$ and $\beta$ such that
$\alpha v \leq w \leq \beta v$. A part is {\em trivial}
if it is reduced to the equivalence class of the zero vector.
{\em Hilbert's projective metric} $d_H$ 
 is defined on every non-trivial part of $V_{\geq 0}$
by the following formula
\[
d_H(v,w) = \log \inf\bigg\{ \frac{\beta}{\alpha} : \alpha,\beta>0, \; \alpha v \leq w \leq \beta v\bigg\}
\enspace .
\]
The infimum is achieved, since $V_{\geq 0}$ is closed.
The map $d_H$ is nonnegative, it satisfies the triangular
inequality, and $d_H(v,w)$ vanishes if, and only if, $v$ and $w$
are proportional -- this justifies the term ``projective metric''.
This metric plays a fundamental role in Perron--Frobenius
theory and in metric geometry,
and also in tropical geometry,
see~\cite{nussbaumlemmens,papadopoulos,cgq02}
for background.

When $V_{\geq 0}=
(\R_{\geq 0})^n$ is the standard orthant, and when
all the entries of the vectors $v$ and $w$ are positive,
we have
\[
d_H(v,w) = \max_{k\in[n]} (\log v_k-\log w_k )-\min_{k\in[n]}(\log v_k-\log w_k )
\enspace. 
\]
Denoting by $e$ the unit vector of $\R^n$, we observe that
\[
d_H(v,w)
= \|\log v -\log w\|_H
\]
where the notation $\log v$ is understood entrywise,
and
\[ \|z\|_H
= 2 \min_{c\in \R} \|x-c e\|_\infty \enspace.
\]
In other words, up to a logarithmic change of variables,
$d_H$ arises by modding out the normed space
$(\R^n,\|\cdot\|_\infty)$ by the one-dimensional
space $\R e$.

We shall suppose that every semigroup $S^i$ has an eigenvector
$u^i\geq 0$, with eigenvalue $\lambda^i$,
meaning that
\[
S^i_t u^i = \exp(\lambda^i t) u^i,\quad \forall t\geq 0 \enspace .
\]
Since $S^i_t $ preserves $V_{\geq 0}$,
this entails that $\lambda^i$ is real.

We choose a linear form $\varphi\in V'$ which
we require to take nonnegative values on $V_{\geq 0}$.
We shall think of $V$ as the
{\em state space} and $\varphi$ as an {\em observable}.
We consider the following scalar observation of the dynamics
\[ Y_t: = \varphi (v_t) \enspace .
\]
We shall assume, in addition, that
$\varphi$ does not vanish on $v_t$, for all $t\geq 0$.
Then, we can define the image
of the observation by the logarithmic map
\[
y_t:= \log Y_t,\qquad \forall t\geq 0 \enspace .
\]

The following result 
shows that the logarithm of the observation
stays at finite distance from a
piecewise linear map.
\begin{theorem}\label{th-1}
  Suppose that the semigroups $S^1,\dots,S^m$ are order
  preserving. Suppose in addition that
  the initial condition $v_0$ and the eigenvectors
  $u^1,\dots, u^m$ all lie in the same non-trivial part of $V_{\geq 0}$,
  and that the linear form $\varphi$ takes positive
  values on this part.
  Then, there exists a constant $\gamma$ such that
  the piecewise linear map $t\mapsto y^\trop_t$ defined,
for $t\in [t_j, t_{j+1})$, by 
\[
y^\trop_t := \lambda_{j+1}(t-t_j)+ \lambda_{j} (t_j-t_{j-1})
+ \dots + \lambda_1 (t_1-t_0)  + \gamma 
\enspace ,
\]
satisfies 
  \[
  |y_t - y^{\trop}_t |\leq \frac{\Delta}{2}, \qquad \forall t\geq 0 \enspace,
  \]
  where
  \[
  \Delta = d_H(v_0,u^1)+ d_H(u^1,u^2)+ \dots +d_H(u^{m-1},u^m)
  \enspace .
  \]
\end{theorem}
\begin{proof}
  By definition of Hilbert's projective metric,
  we can find positive constants $\alpha_0,\beta_0$,
  such that $\alpha_0 u^1\leq v_0 \leq \beta_0 u^1$
  and $d_H(v_0,u^1)=\log(\beta_0/\alpha_0)$.
  Similarly, for all $i\in[m-1]$, we can find
  positive constants $\alpha_i,\beta_i$,
  such that $\alpha_{i} u^{i+1}\leq u^{i} \leq \beta_{i} u^{i+1}$,
  and $d_H(u^{i},u^{i+1})=\log(\beta_i/\alpha_i)$.
  For all $j\geq 0$, with $j\leq m-1$,
 and for all $t\in [t_j,t_{j+1})$, we set
  \[
  z_t:=  \lambda_{j+1}(t-t_j) + \lambda_{j}(t_j-t_{j-1})
  +\dots + \lambda_1(t_1-t_0) \enspace. 
  \]
Since the semigroups $S^i$ are linear and order preserving, 
we prove by induction 
  \begin{align*}
    \exp(z_t )\alpha_j\dots\alpha_0 u^{j+1} 
    \leq v_t \leq \exp(z_t)\beta_j\dots\beta_0 u^{j+1} \enspace .
  \end{align*}
  We observe that
  \[
\alpha_{m-1}\dots \alpha_{j+1} u^m
  \leq 
  u^{j+1} \leq \beta_{m-1}\dots \beta_{j+1} u^m
  \]
and so
    \begin{align*}
    \exp(z_t )\alpha_{m-1}\dots\alpha_0 u^{m} 
    \leq v_t \leq \exp(z_t)\beta_{m-1}\dots\beta_0 u^{m} \enspace .
    \end{align*}
  Applying the linear form $\varphi$ to latter inequalities,
  taking the image by the log map,
  and setting
  \[
  \gamma:= \log \varphi(u^m)  + \frac{1}{2}\sum_{j=0}^{m-1}\log(\beta_j \alpha_j)  \enspace, 
  \]
  we arrive at the bound of the theorem.
\end{proof}
A general principle from tropical geometry
states that using ``logarithmic glasses''
reveals a piecewise linear structure~\cite{viro,itenberg}.
\Cref{th-1} is inspired by this principle.
This motivates the notation $y^\trop$,
for the ``tropicalization'' of the logarithm
of the observable $Y$.

\begin{remark}\label{rk-exact}
  If the space $V$ is of dimension $1$, then the bound $\Delta$
  appearing in~\Cref{th-1} is zero, implying that the approximation
  of the logarithm of observables by a piecewise linear curve
  is exact. This occurs if one considers a SIR ODE model:
  then, in the early stage of the epidemics, the dynamics can be
  written only in terms of the population of the one-dimensional I compartment.
\end{remark}
\begin{remark}
\Cref{th-1} carries over to discrete time systems
in a straightforward manner.
\end{remark}
\subsection{Application to the transport PDE model}
\Cref{th-1} applies in particular to the transport model~\eqref{e-transportsimple-compact}. 
Then, as noted above,
the evolution operator of the system~\eqref{e-transportsimple}
preserves the space $V=L^1([0,x^*_E])\times L^1([0,x^*_I])$.
Moreover, when the epidemiological control
term $\control(t)$ is constant, \Cref{prop-eig} shows that
the eigenproblem~\eqref{e-transportsimpleeig} has a
positive and continuous solution $\bar{n}$, with
a real eigenvalue $\lambda$. Different 
stages of sanitary policies correspond to successive
values $\control^1,\ldots,\control^m$
of $\control(t)$, leading to different semigroups $S^i$, $i\in [m]$.
Then, the solution $v_t:= n(\cdot,t)$ of~\eqref{e-transportsimple}
is determined as in~\eqref{evolPerron}.
Each semigroup $S^i$ yields a continuous and positive eigenvector
$u^i:= \bar{n}^i$
satisfying~\eqref{e-transportsimpleeig} associated
with a real eigenvalue $\lambda^i$ of $S^i$.
Two continuous and positive
functions defined on a compact interval are always
in the same part of the cone of nonnegative
functions of $V$, so \Cref{th-1} applies
to this model.

We next give an explicit estimate for the Hilbert projective
distances between eigenvectors, arising in \Cref{th-1}.
\begin{proposition}\label{prop-bound}
  Suppose that Assumption~\ref{as-nonden} holds,
  and that
  for $i=1,2$, 
  $(\lambda^i,\bar{n}^i)$ is the solution  $(\lambda,\bar{n})$
  of the Perron-Frobenius eigenproblem~\eqref{e-transportsimpleeig} when  $\control=\control^i$.
Then,  we have 
 \[ d_H(\bar{n}^1, \bar{n}^{2})\leq |\lambda_1-\lambda_{2}| (x^*_E+x^*_I)
 \enspace .\]
\end{proposition}
The term $x^*_E+x^*_I$ is the maximal time
elapsed between contamination and the end of infectiosity.
\begin{proof}[Proof of \Cref{prop-bound}]
  Suppose, without loss of generality, that $\bar{n}^i_E(0)=1$
  for $i=1,2$.
  Let $F^\lambda_{E\to I}$ and $F^\lambda_{I\to R}$ be defined as in~\eqref{defFlambda}. We have $F^\lambda_{E\to I}(x)=\lambda x+  F^0_{E\to I}(x)$ and
  $F^\lambda_{I\to R}(x)= \lambda x+  F^0_{I\to R}(x)$.
  Then, \eqref{e-explicit} and the boundary condition in~\eqref{e-eig2}
  yield
  \begin{align}
  \bar n_E^i(x)&= \exp(-\lambda^i x)\exp\big( - F^0_{E\to I}(x)\big) \,,
    \text{ for } 0\leq x \leq x_E^*\,,\label{e-tob1}\\
    \bar n^i_I(0)& = \int_0^{x_E^*} K_{E\to I}(x)n^i_E(x)\dd x +  n_E^i(x_E^*)\,,\label{e-tob2}\\
\bar n_I^i(y)&= \exp(-\lambda^i y) \exp\big( - F^0_{I\to R}(y)\big) \bar n_I^i(0)\,,
    \text{ for } 0\leq y\leq  x_I^*\,.\label{e-tob3}
  \end{align}
  Let $j\in \{1,2\}$ be distinct from $i$, and set $t^+:=\max(t,0)$.
  Bounding $\exp(-\lambda^i x)$
   by $\exp(-\lambda^j x) \exp\big( (\lambda^j-\lambda^i)^+ x_E^*\big)$
   in~\eqref{e-tob1},  we obtain:
   \[   \bar n_E^i(x)\leq \exp\big( (\lambda^j-\lambda^i)^+ x_E^*\big) \bar n_E^j
   (x)\,, \text{ for } 0\leq x \leq x_E^*\enspace .\]
   Applying this inequality in~\eqref{e-tob2}, we deduce
   \[   \bar n_I^i(0)\leq \exp\big( (\lambda^j-\lambda^i)^+ x_E^*\big) \bar n_I^j (0)\enspace .\]
   Now applying this inequality and bounding $\exp(-\lambda^i y)$
   by $\exp(-\lambda^j y) \exp\big( (\lambda^j-\lambda^i)^+ x_I^*\big)$
   in~\eqref{e-tob3},
   we obtain:
   \begin{equation}\label{ine-nI}
     \bar n_I^i(y)\leq \exp\big( (\lambda^j-\lambda^i)^+ (x_E^*+x_I^*)\big) \bar n_I^j (y)\,,  \text{ for } 0\leq y\leq  x_I^*\enspace .\end{equation}
  So,
  \[
  d_H( (\bar n_E^i,\bar n_I^i), (\bar n_E^j,\bar n_I^j))
  \leq (\lambda^j-\lambda^i)^+ (x_I^*+x_E^*)+
  (\lambda^i-\lambda^j)^+ (x_I^*+x_E^*)
  = |\lambda^i-\lambda^j| (x_I^*+x_E^*) \, .\qedhere
  \]
\end{proof}

   Note that applying \eqref{ine-nI} to the boundary condition in~\eqref{e-Lotkaeig}, we also deduce from the previous proof that 
  \[
  \bar{n}_E^i (0) \leq \frac{\control^i}{\control^j}\exp\big( (\lambda^j-\lambda^i)^+(x_I^*+x_E^*)\big) \bar n_E^{j}(0)\,.
  \]
  Since   $\bar{n}_E^i (0)=1$ for $i=1,2$, we get the following bound
  \[
  \control^j/\control^i \leq \exp\big((\lambda^j-\lambda^i)^+\big)(x_I^*+x_E^*)
  \,,
  \]
  from which one recover that $\control$ is a nondecreasing function of $\lambda$, and which gives an estimation of the rate $\control$ in terms of the eigenvalue $\lambda$. 
  
  The bound of~\Cref{th-1}
may be refined. This is left for further work.

\section{Short term predictions}\label{sec-appendixC}\label{sec-proba}
We now describe the basic methodology we propose to build confidence intervals for future occurrences of medical events related to epidemic progression,  and raise alarms about its potential resurgence. We first consider a single time series of numbers of event occurrences. We then describe how to consolidate several time series corresponding to distinct medical events in order to construct improved alarm criteria. The simpler case of least squares fitting is considered first, the more robust $\ell_1$ alternative is described next.

\paragraph{Time series for a single type of events:} Let $X(1),\ldots,X(n)$ be indices of days, and we aim to do a forecast based on observations made on these days.
 Typically, on day $d_0$, we may select $n=7$ and let $X(1)=d_0-n,\ldots,X(n)=d_0-1$ to perform a forecast on the basis of the last seven days. Let $Y(t)$ denote the count of medical events (for instance, dispatches of ambulances) on day $X(t)$, and let $Z(t)=\log Y(t)$. Based on the previous discussion (epidemiological modeling) we assume that for all $t=1,\ldots,n$,
$$
Z(t)=\alpha + \beta X(t) +\epsilon_t
$$
for constants $\alpha$, $\beta$, where $\epsilon_t$ denotes some random noise. For simplicity we assume here i.i.d.\ noise sequence $\epsilon_1,\ldots,\epsilon_n$, and that each $\epsilon_t$ admits a Gaussian distribution $\cN(0,\sigma^2)$ with zero mean and variance $\sigma^2$. 

Least-square estimates for the parameters $\alpha$, $\beta$ are then provided by
\begin{equation}
\hat{\beta}=\frac{\sum_{t=1}^n (X(t)-\bar{X})(Z(t)-\bar{Z})}{\sum_{t=1}^n(X(t)-\bar{X})^2},\quad \hat{\alpha}=\bar{Z}-\hat{\beta}\bar{X},
\end{equation}
where 
\begin{equation}
\bar{X}=\frac{1}{n}\sum_{t=1}^nX(t),\quad \bar{Z}=\frac{1}{n}\sum_{t=1}^nZ(t).
\end{equation}
The variance $\sigma^2$ can be estimated as
\begin{equation}
\hat{\sigma}^2=\frac{1}{n-2}\sum_{t=1}^n(Z(t)-\hat{Z}(t))^2,
\end{equation}
where 
\begin{equation}
\hat{Z}(t):=\hat{\alpha}+\hat{\beta}X(t).
\end{equation}
Under the assumptions of i.i.d.\ Gaussian errors $\epsilon_t$, we have that, for each $t$ corresponding to a future day $X(t)$ (in particular, $t\notin\{1,\ldots,n\}$), the three following variables:
$$
\displaystyle \frac{\hat{\alpha}-\alpha}{\hat{\sigma}\sqrt{\frac{1}{n}+\frac{\bar{X}^2}{\sum_{t'=1}^n (X(t')-\bar{X})^2}}},\;
\frac{\hat{\beta}-\beta}{\hat{\sigma}\sqrt{\frac{1}{\sum_{t'=1}^n(X(t')-\bar{X})^2}}},\;\;,
\frac{Z(t)-\hat{Z}(t)}{\hat{\sigma}\sqrt{1+\frac{1}{n}+\frac{(X(t)-\bar{X})^2}{\sum_{t'=1}^n(X(t')-\bar{X})^2}}},
$$
all admit a bilateral Student distribution with $n-2$ degrees of freedom (see~\cite[Ch. 28]{johnson1994continuous} or~\cite{cramer1999mathematical}). Denote by $t^{n-2}_{\gamma}$ the $\gamma$-th quantile of this distribution. For $\epsilon \in [0,1]$, this provides us with the following confidence intervals with confidence $1-\epsilon$:
\begin{equation}
\begin{array}{lll}
\alpha&\in& \left[\hat{\alpha}-\hat{\sigma}\sqrt{\frac{1}{n}+\frac{\bar{X}^2}{\sum_{t'=1}^n (X(t')-\bar{X})^2}} t^{n-2}_{1-\epsilon/2},\hat{\alpha}+\hat{\sigma}\sqrt{\frac{1}{n}+\frac{\bar{X}^2}{\sum_{t'=1}^n (X(t')-\bar{X})^2}} t^{n-2}_{1-\epsilon/2}\right],\\
\beta &\in& \left[\hat{\beta}-\hat{\sigma}\sqrt{\frac{1}{\sum_{t'=1}^n(X(t')-\bar{X})^2}}t^{n-2}_{1-\epsilon/2},\hat{\beta}+\hat{\sigma}\sqrt{\frac{1}{\sum_{t'=1}^n(X(t')-\bar{X})^2}}t^{n-2}_{1-\epsilon/2} \right]\\
Z(t)&\in& \left[\hat{Z}(t)-\hat{\sigma}\sqrt{1+\frac{1}{n}+\frac{(X(t)-\bar{X})^2}{\sum_{t'=1}^n(X_(t')-\bar{X})^2}}t^{n-2}_{1-\epsilon/2}, \hat{Z}(t)+\hat{\sigma}\sqrt{1+\frac{1}{n}+\frac{(X(t)-\bar{X})^2}{\sum_{t'=1}^n(X(t')-\bar{X})^2}}t^{n-2}_{1-\epsilon/2} \right]
\end{array}
\end{equation}
As an illustration, for $n=7$ and $\epsilon=5\%$, we can plug in $t^{5}_{0.975}=2.571$ in the last interval, and thus obtain a $95\%$-confidence interval centered around $\hat{Z}(t)$ for $Z(t)=\log Y(t)$, the logarithm of the count $Y(t)$ on a future day $X(t)$, that is:
\begin{equation}
Z(t)\in \left[\hat{Y}(t)-2.571\times\hat{\sigma}\sqrt{1+\frac{1}{n}+\frac{(X(t)-\bar{X})^2}{\sum_{t'=1}^n(X(t')-\bar{X})^2}}, \hat{Z}(t)+2.571\times\hat{\sigma}\sqrt{1+\frac{1}{n}+\frac{(X(t)-\bar{X})^2}{\sum_{t'=1}^n(X(t')-\bar{X})^2}}\right]
\end{equation}

Although we could extend this definition of the confidence interval for the short terms predictions of the value of $Z(t)$, we propose a more conservative confidence domain, in the shape of a trapezoid. It is obtained by extending the upper-bound $Z(t)^+$ (resp.\ the lower bound $Z(t)^-$ of the $95\%$ confidence interval on $Z(t)$ by a line with slope equal to the upper-bound $\beta^+$ (resp. lower-bound $\beta^-$) of $95\%$ confidence interval on $\beta$. For a given day $t$, the upper and lower envelopes of the trapezoid have ordinates

\[      \bigg(\hat{\beta}(t-t_n)+\hat{Z}_n\bigg)\pm \displaystyle\left(\sqrt{\mathrm{Var}(\hat{\beta})}(t-t_n)+\sqrt{\widehat{\sigma}^2+\mathrm{Var}(\hat{Z}_n)}\right) t^{n-2}_{1-\epsilon/2}\enspace.
\]

If instead of the count $Y(t)$ on a particular day $X(t)$, we are interested in the trend of the epidemic, whether exploding or contracting, we should then consider the confidence interval for parameter $\beta$. Again for $n=7$ and $\epsilon=5\%$ this gives
\begin{equation}
\beta \in \left[\hat{\beta}-2.571\times\hat{\sigma}\sqrt{\frac{1}{\sum_{t'=1}^n(X(t')-\bar{X})^2}},\hat{\beta}+2.571\times\hat{\sigma}\sqrt{\frac{1}{\sum_{t'=1}^n(X(t')-\bar{X})^2}} \right]\\
\end{equation}
One-sided confidence intervals may also be provided, and are in fact more natural for the definition of alarm indicators. 

For concreteness, assume we want to raise an alarm when the doubling time, $\delta=(\log 2)/\beta$, is $\delta^*$ days or less, where $\delta^*$ could be 10 for instance. This is equivalent to $\beta$ exceeding $(\log 2)/\delta^*$. Thus $\delta$ is less than $\delta^*$ days with confidence $1-\epsilon$ when 
$$
\frac{\log 2}{\delta^*} < \hat{\beta}- t^{n-2}_{1-\epsilon} \sqrt{V},
$$
where 
$$
V=\hat{\sigma}\sqrt{\frac{1}{\sum_{t'=1}^n(X(t')-\bar{X})^2}}.
$$
Raising an alarm under this condition then amounts to calibrating the false positive probability at $\epsilon$. For instance, for $\epsilon = 5\%$, and $n=7$, we would plug in $t^7_{0.95}=2.015$ in the above expression.

Alternatively, raising an alarm under the condition
$$
\frac{\log 2}{\delta^*} < \hat{\beta}+t^{n-2}_{1-\epsilon} \sqrt{V},
$$
 corresponds to calibrating the false negative probability (probability of not raising an alarm while $\delta\le 10$) at $\epsilon$.

Our alarm indicators correspond to the first choice, i.e. calibration of a false positive rate, with $\delta^*$ set to $+\infty$.

\paragraph{Alarm indicators based on multiple types of events:} Assume that several types $j$ of events are available, and let $J$ denote the corresponding set of events. For instance, we could distinguish between dispatches of ambulances bringing patients to Intensive Care Units as opposed to Non-intensive Care Units, thereby producing two distinct time series. Let $X_j(t)$, $t=1,\ldots,n_j$ denote the days on which counts $Y_j(t)$ of type $j$ event occurrences are to be used. Let $Z_j(t)=\log Y_j(t)$. We assume as before the linear regression model
$$
Z_j(t)= \alpha_j +\beta_j X_j(t)+\epsilon_j(t),\; t=1,\ldots,n^j \enspace.
$$
Now for each of these times series, we can produce, based on the previous discussion, the estimator 
$$
\hat{\beta}_j:=\frac{\sum_{t=1}^{n_j} (X_j(t)-\bar{X}_j)(Z_j(t)-\bar{Z}_j)}{\sum_{t=1}^{n_j}(X_j(t)-\bar{X}_j)^2},
$$
where 
$$
\bar{X}_j=\frac{1}{n_j}\sum_{t=1}^{n_j}X_j(t),\quad \bar{Z}_j=\frac{1}{n_j}\sum_{t=1}^{n_j}Z_j(t).
$$
Suppose in addition that the noise terms $\epsilon_j(t)$ are mutually independent, Gaussian, with zero mean and variance $\sigma_j^2$ for errors $\epsilon_j(t)$. 
Suppose finally that the exponents $\beta_j$ all coincide with $\beta$, the exponent that is characteristic of the epidemic's progression. Denote by 
\begin{equation}
V_j:=\hat{\sigma}_j^2\sqrt{\frac{1}{\sum_{t=1}^{n_j}(X_j(t)-\bar{X}_j)^2}},
\end{equation}
where, reproducing the computations for a single time series, we let
$$
\hat{\sigma}_j^2:=\frac{1}{n_j-2}\sum_{t=1}^{n_j}(Z_j(t)-\hat{Z}_j(t))^2,
$$
and
$$
\hat{Z}_j(t):=\hat{\alpha}_j+\hat{\beta}_jX_j(t).
$$
As previously, $V_j$ is our estimate of the variance of estimate $\hat{\beta}_j$. We finally propose to combine the individual estimators $\hat{\beta}_j$ into
\begin{equation}
\hat{\beta}:=\frac{\sum_{j\in J} \frac{1}{V_j}\hat{\beta}_j}{\sum_{j\in J}\frac{1}{V_j}}\cdot
\end{equation}
For the sake of simplicity, let us approximate the bilateral Student distribution with $n-2$ degrees of freedom by the standard distribution $\cN(0,1)$. We then have the approximate distributions $\hat{\beta}_j\approx \cN(\beta, V_j)$, and hence the approximate distribution $\hat{\beta}-\beta \approx \cN(0,V)$,
\begin{equation}
V:=\frac{1}{\sum_{j\in J} \frac{1}{V_j}}\cdot
\end{equation}
Weighing  the individual estimators $\hat \beta _j$ by the reciprocal of their variances as just done minimizes the variance of the resulting estimator. The same approach as previously considered then leads to the following conditions for alarm raising:

To raise an alarm when the doubling time $\delta=(\log 2)/\beta$ exceeds $\delta^*$ days (e.g., $\delta^*=10$), if we target a false alarm probability of $\epsilon$, we are led to raise an alarm when Condition
\begin{equation}
\frac{\log 2}{\delta^*}< \hat{\beta} -g_{1-\epsilon}\sqrt{V},
\end{equation}
where $g_{1-\epsilon}$ is the $1-\epsilon$-quantile of the standard Gaussian distribution.

If instead we target a false negative probability (probability of not raising an alarm) at $\epsilon$, we would then raise an alarm when 
\begin{equation}
\frac{\log 2}{\delta^*}< \hat{\beta}+g_{1-\epsilon} \sqrt{V},
\end{equation}

\paragraph{More robust $\ell_1$-based approach:} The previous estimators and derived alarm conditions have the appeal of simplicity, but can be advantageously replaced by more robust versions, that are less sensitive to the presence of outliers. 

A popular alternative is the following $\ell_1$ criterion. We again consider $Z_j(t):=\log Y_j(t)$, where $Y_j(t)$ is the number of type $j$ events on day $X_j(t)$. We then let $\hat\alpha_j$, $\hat\beta_j$ achieve the minimum of the criterion $\sum_{t=1}^{n^j} |\alpha+\beta X^j_t -Y^j_t|$. 
They are obtained by solving a linear program.
Here we assume that observations $Z_j(t)$ are mutually independent and distributed according to density 
$f_{j,t}(z)=\frac{1}{2\lambda^j}\exp(-|z-\alpha_j-\beta_j X_j(t)|/\lambda_j)$. In other words this corresponds to adding a Laplace observation noise with density $\frac{1}{2\lambda^j}\exp(-|z|/\lambda_j)$ to the signal of interest $\alpha_j+X_j(t)\beta_j$. The above $\ell_1$ minimization criterion corresponds to maximum likelihood estimation of $\alpha_j$, $\beta_j$ in this observational noise model, as its log-likelihood is given by
$$
-|T|\log (2 \lambda_j) -\sum_{t=1}^{n_j}\frac{|Z_j(t)-\alpha_j-\beta_j X_j(t)|}{\lambda_j}.
$$

A rich theory for the performance of the resulting estimators is available, see for instance~\cite{bassett-koenker}. The latter work treats general i.i.d.\ errors, and do not restrict itself to e.g.\ Laplacian distribution of errors; recent work like~\cite{stangenhaus1993bootstrap} experiments techniques to obtain confidence intervals when distribution of errors is unknown. Here we make the choice of Laplace-distributed errors for sake of simplicity.
In particular, the asymptotic theory in \cite{bassett-koenker} suggests the approximation 
$$
\hat{\beta}_j\sim \cN( \beta_j,V_j)
$$ where 
\begin{equation}
\hat{\lambda}_j:=\frac{1}{n_j}\sum_{t=1}^{n_j}|Z_j(t)-\hat\alpha_j-\hat\beta_j X_j(t)|,
\end{equation}
and
\begin{equation}
 V_j:=(\hat\lambda_j)^2\frac{1}{\sum_{t=1}^{n_j}X_j(t)^2 -\frac{1}{n_j}(\sum_{t=1}^{n_j} X_j(t))^2}\cdot
\end{equation}
We again consider that multiple types $j\in J$ of time series are conjointly available, and that each $\beta^j$ coincides with $\beta$, the parameter to be estimated.
Assuming the $\hat\beta_j$ to be independent with $\hat\beta\sim \cN( \beta_j,V_j)$, leads us to define the estimator
\begin{equation}
\hat\beta:=\frac{\sum_{j\in J}\frac{\hat\beta_j}{(\hat{\lambda}_j)^2}}{\sum_{j\in J}\frac{1}{(\hat{\lambda}_j)^2}},
\end{equation}
whose distribution is then given by $\hat{\beta}\sim \cN(\beta, V)$ where
\begin{equation}
V=\frac{1}{\sum_{j\in J}\frac{\sum_{t=1}^{n_j} X_j(t)^2 -(\sum_{t=1}^{n_j}X_j(t))^2/n_j}{(\hat{\lambda}_j)^2} }\cdot
\end{equation}
A symmetric $(1-\epsilon)$-confidence interval for $\beta$ is then provided by
\begin{equation}
\beta \in I:=\left[\hat{\beta}-g_{1-\epsilon/2} \sqrt{V},\hat{\beta}+g_{1-\epsilon/2} \sqrt{V}\right].
\end{equation}
Similarly, $(1-\epsilon)$-confidence one-sided intervals for $\beta$ are obtained by letting
\begin{equation}
\beta \in I':=[\hat{\beta}-g_{1-\epsilon} \sqrt{V},+\infty), \; \beta \in I'':=(-\infty,\hat{\beta}+g_{1-\epsilon}\sqrt{V}].
\end{equation}
The doubling time $\delta$ is given by $(\log 2)/\beta$ if $\beta>0$, and $+\infty$ otherwise. This gives the $1-\epsilon$-confidence conditions for $\delta$:
\begin{equation}
\hbox{if } \hat{\beta}-g_{1-\epsilon} \sqrt{V}>0,\; \delta\in I_1=\left[0,\frac{\log 2}{\hat{\beta}-g_{1-\epsilon} \sqrt{V}}\right],
\end{equation}
and
\begin{equation}
\delta\in I_2=\left[\frac{\log 2}{\max(0,\hat{\beta}+g_{1-\epsilon} \sqrt{V})},+\infty\right).
\end{equation}
For concreteness assume we want to raise an alarm when $\delta$ is $\delta^*$ days or less, where $\delta^*$ could be 10. 
From the above consideration, $\delta$ is below $\delta^*$ days with confidence $1-\epsilon$ when 
$$
\frac{\log 2}{\delta^*} < \hat{\beta}-g_{1-\epsilon} \sqrt{V}.
$$
Raising an alarm under this condition then amounts to calibrating the false positive probability at $\epsilon$. 

Alternatively, we may consider to raise an alarm under the condition
$$
\frac{\log 2}{\delta^*} < \hat{\beta}+g_{1-\epsilon} \sqrt{V}.
$$
This would correspond to calibrating the false negative probability (probability of not raising an alarm while $\delta\le \delta^*$) at $\epsilon$.

\section{Conclusion}
We have shown that monitoring of emergency calls to EMS allows to anticipate
the evolution of an epidemic by providing several {\em early signals},  each with specific characteristics in terms of time lag and  reliability.

Our study illustrates the spatially differentiated nature of the epidemic kinetics, with significant doubling time differences between neighboring departments. 

Such spatial differentiation, if present at a granularity finer than that of departments considered here, could be exploited using the methods described in the present work in order to detect potential epidemic resurgences at the corresponding spatial granularity. This shows great promise in enabling detection of so-called epidemic clusters.

There is thus huge potential in the extension of this work and its application to finer spatial resolution. 

Notwithstanding such extensions, monitoring epidemic kinetics through EMS calls at regional levels can already be exploited to define region-specific sanitary measures, such as lifting of travel bans, proportionate to the regional situation, and to allow early detection of epidemic resurgence. Importantly, we expect this finding to be applicable in full generality to EMS organizations worldwide. Thus the methods introduced here may be of wide applicability to combat \covid-19. Beyond \covid-19, EMS organizations have a unique role to play in early detection of sanitary crises. 

\section{Acknowledgments}
We thank the operational team of DSI of AP-HP, who helped to extract information records, especially Stéphane Crézé, Laurent Fontaine,
Pierre Cabot, François Planeix, Fabrice Tordjman, Grégory Terrell
and Martine Spiegelmann.

We thank Pr.\ Renaud Piarroux for very helpful remarks.
We thank Pr.\ Bruno Riou for his suggestion to include
quantitative statistical estimates in the present article.
We thank Pr.\ Frédéric Batteux for having provided
epidemiological information.
We thank Dr.\ François Braun (SAMU 57) and Dr.\ Vincent Bounes (SAMU 31) for providing comparison elements between their departments. We thank Dr.\ Nicolas Poirot for introducing us to SAMU 31.
We thank Dr.\ Paul-Georges Reuter (SAMU 92) for useful comments on the interpretation of SAMU data relative to the Covid crisis.

We thank Ayoub Foussoul, for having developed
a robust dynamic programming algorithm, allowing
one to consolidate the results of this manuscript
concerning the best piecewise linear approximation
of the log of observables.
We thank Jérôme Bolte, for providing insights
on non-convex and non-smooth best-approximation problems.

We thank Tania Lasisz for her help in the administration of the project, and Guillermo Andrade Barroso, Thomas Calmant and Matthieu Simonin for their contribution to software development.

We thank NXO France Integrator of communication solutions team and SIS Centaure15 solution from GFI World team for the help they provided and
their availability for the project.

We thank Orange Flux Vision (especially Jean-Michel Contet)
for having provided
daily population estimates, at the scale of the department,
helping to calibrate our models.

We thank Enedis (especially Pierre Gotelaere and his team) for having
provided an estimation of the departure rate
of households, aggregated at the scale of departments and districts,
helping us to refine our model.

We thank SFR Geostatistic Team (especially Loic Lelièvre)
for having provided estimates of flows
between Paris and province, aggregated at the scale
of departments and districts, allowing us
to incorporate mobility in our model.

Stéphane Gaubert thanks Nicolas Bacaër for a decisive
help, concerning epidemiological and mathematical analysis,
provided during the week of March 16th-20th.
He thanks Cormac Walsh for improvements of the text. He also
thanks Thomas Lepoutre for very helpful mathematical comments and suggestions
concerning~\Cref{sec-pde}.

The INRIA--École polytechnique team thanks the Direction de Programme de la Plate Forme d'Appels
d'Urgences -- PFAU at Préfecture de Police, DOSTL (Régis Reboul), and Brigade de Sapeurs
Pompiers de Paris (especially Gen.\ Jean-Marie Gontier and Capt.\ Denis Daviaud)
for having provided precious elements of comparison concerning the calls
received at the emergency numbers 17-18-112.

\nocite{bacaer}
%% \bibliographystyle{plain}
%% \bibliography{references}

\appendix
\section{Appendix: algorithms to compute a best approximation of the logarithm of the number of events by a piecewise linear map}
\label{appendix-2}
Given an epidemiologic observable $Y(t)$, we need
to approximate $\log Y(t)$ by a function
\[ \linear(t):=\min_{1\leq j\leq \nu}
(\lambda_j  t  + c_j)\,,
\]
where $\nu$ is the number of phases with constant sanitary policy during the considered time period. The parameters $\lambda_j$, $c_j$ are assumed without loss of generality to satisfy $\lambda_1\le\lambda_2\cdots\le \lambda_{\nu}$. The concavity constraint imposed on the approximating function $\linear(t)$ makes the problem different from standard function approximation problems, and contributes to the robustness of the fitting procedure by reducing the amount of overfitting.

The two most natural criteria for fitting function $\linear(t)$ to observations $\log Y(t)$ are to minimize either a least squares, or $\ell_2$ loss function $\sum_{t\in \mathcal{T}} |\linear(t)-\log Y(t)|^2$, or an $\ell_1$ loss function $\sum_{t\in\mathcal{T}} |\linear(t)-\log Y(t)|$, where $\mathcal{T}$ is a finite
set of time instants at which observations have been made. As discussed in \Cref{sec-proba}, the $\ell_1$ formulation is more robust in being less sensitive to outliers, and is the one used on~\Cref{p-phases}. 

The corresponding optimization problem over parameters $\lambda_i$, $c_i$ is non-convex as soon as $\nu \ge 2$. A straightforward option is to use a derivative free procedure,
like the Nelder-Mead~\cite{lagarias} algorithm. Depending on the initial
point, this algorithm
may converge to a local minimum, which may not be epidemiologically
significant. So, a possibility is to guide the algorithm
by providing it a initial guess of the optimal solution.
To do, we start by an a priori selection of the time periods over which function $\linear(t)$ is linear (which could be obtained by prior knowledge of delay parameters $\tau$ and times of policy changes, or found by brute force search). We then determine a minimum cost linear fit of target function $\log Y(t)$ over each such period, and use the concave envelope of the resulting function as our initial condition for local search.  This is how we initially obtained
the best $\ell_1$ approximation shown on~\Cref{p-phases}. We also
used CMA-ES for comparison~\cite{hansen}. Both Nelder-Mead
and CMA-ES algorithms appear to be sensitive to the initial conditions.
Notice in this respect that the objective function is linear on the cells
of a polyhedral complex and that it can be constant on certain unbounded cells
of this complex, so a local search
algorithm may be trapped in a cell in which the function is constant.
Another perspective is to observe that this best approximation
problem is equivalent to a learning problem, looking for the parameters of a neural networks with a single hidden layer and min-type activation functions, see~\cite{cala}. This allows one to apply (nonsmooth) optimization algorithms used
in learning, still leading in general to a local optimum.
An approach leading to the global optimum is dynamic programming,
originating from Bellman~\cite{bellman}. Ayoub Foussoul (École polytechnique)
provided us with a dynamic programming solver, implementing several refinements,
and allowing us to certify the global optimality of the
approximation shown in~\Cref{p-phases},
up to a fixed precision.

\end{document}